\documentclass[sigplan,screen]{acmart}

\usepackage{microtype}

\usepackage{amsmath}
\usepackage{amsthm}
\usepackage{cmll}
\usepackage{braket}
\usepackage{proof}
\usepackage[inline,shortlabels]{enumitem}
\usepackage{listings}
\usepackage{cleveref}
\usepackage{preamble}

\newcommand{\basicstylesize}{\small}
\newcommand{\basicstylespread}{\linespread{0.8}}
\newcommand{\basicstylefamily}{\ttfamily}
\newcommand{\basicstyle}[1][]{\basicstylespread\basicstylesize\basicstylefamily #1}
\setcopyright{cc}
\setcctype{by}
\acmDOI{10.1145/3759427.3760362}
\acmYear{2025}
\copyrightyear{2025}
\acmISBN{979-8-4007-2150-2/25/10}
\acmConference[OLIVIERFEST '25]{Proceedings of the Workshop Dedicated to Olivier Danvy on the Occasion of His 64th Birthday}{October 12--18, 2025}{Singapore, Singapore}
\acmBooktitle{Proceedings of the Workshop Dedicated to Olivier Danvy on the Occasion of His 64th Birthday (OLIVIERFEST '25), October 12--18, 2025, Singapore, Singapore}
\received{2025-06-10}
\received[accepted]{2025-07-31}

\begin{document}

%%
%% The "title" command has an optional parameter,
%% allowing the author to define a "short title" to be used in page headers.
\title{Controlling Copatterns: There and Back Again}
\subtitle{Extended Version}

%%
%% The "author" command and its associated commands are used to define
%% the authors and their affiliations.
%% Of note is the shared affiliation of the first two authors, and the
%% "authornote" and "authornotemark" commands
%% used to denote shared contribution to the research.
\author{Paul Downen}
\orcid{0000-0003-0165-9387}
\affiliation{%
  \institution{University of Massachusetts Lowell}
  \city{Lowell}
  \country{USA}
}
\email{Paul_Downen@uml.edu}

%%
%% By default, the full list of authors will be used in the page
%% headers. Often, this list is too long, and will overlap
%% other information printed in the page headers. This command allows
%% the author to define a more concise list
%% of authors' names for this purpose.
\renewcommand{\shortauthors}{Downen}

%%
%% The abstract is a short summary of the work to be presented in the
%% article.
\begin{abstract}
  Copatterns give functional programs a flexible mechanism for responding to
  their context, and composition can greatly enhance their expressiveness.
  However, that same expressive power makes it harder to precisely specify the
  behavior of programs.  Using Danvy's functional and syntactic correspondence
  between different semantic artifacts, we derive a full suite of semantics for
  copatterns, twice.  First, a calculus of monolithic copatterns is taken on a
  journey from small-step operational semantics to abstract machine to
  continuation-passing style.  Then within continuation-passing style, we
  refactor the semantics to derive a more general calculus of compositional
  copatterns, and take the return journey back to derive the other semantic
  artifacts in reverse order.
\end{abstract}

%%
%% The code below is generated by the tool at http://dl.acm.org/ccs.cfm.
%% Please copy and paste the code instead of the example below.
%%

% \ccsdesc[500]{Software and its engineering~Semantics}

%%
%% Keywords. The author(s) should pick words that accurately describe
%% the work being presented. Separate the keywords with commas.
%% A "teaser" image appears between the author and affiliation
%% information and the body of the document, and typically spans the
%% page.
% \begin{teaserfigure}
%   \includegraphics[width=\textwidth]{sampleteaser}
%   \caption{Seattle Mariners at Spring Training, 2010.}
%   \Description{Enjoying the baseball game from the third-base
%   seats. Ichiro Suzuki preparing to bat.}
%   \label{fig:teaser}
% \end{teaserfigure}

% \received{10 June 2025}
% \received[revised]{16 June 2025}
% \received[accepted]{14 August 2025}

%%
%% This command processes the author and affiliation and title
%% information and builds the first part of the formatted document.
\maketitle

\theoremstyle{acmdefinition}
\newtheorem{remark}[theorem]{Remark}

\crefname{section}{section}{sections}
\Crefname{section}{Section}{Sections}
\crefname{figure}{figure}{figures}
\Crefname{figure}{Figure}{Figures}
\crefname{equation}{equation}{equations}
\Crefname{equation}{Equation}{Equations}

\crefname{theorem}{theorem}{theorems}
\Crefname{theorem}{Theorem}{Theorems}
\crefname{lemma}{lemma}{lemmas}
\Crefname{lemma}{Lemma}{Lemmas}
\crefname{corollary}{corollary}{corollaries}
\Crefname{corollary}{Corollary}{Corollaries}

\crefname{example}{example}{examples}
\crefname{example}{Example}{Examples}
\crefname{remark}{remark}{remarks}
\Crefname{remark}{Remark}{Remarks}

\section{Introduction}
\label{sec:intro}

Pattern matching---a common feature among functional languages that expresses
complex traversals over trees---can be made even more powerful through a modern
extension: \emph{copatterns} \cite{Copatterns}.  The dual to patterns,
copatterns let multi-clause definitions match over more information in their
calling context, reacting to the structure of projections in addition to the
structure of arguments.  In contrast to Haskell-style lazy data structures,
copatterns are especially useful for modeling and manipulating infinite
information like streams in settings that are more sensitive to termination like
proof assistants and eager languages.

\emph{Compositional copatterns}---as recently implemented as a Scheme macro
library \cite{CoScheme}---extend copatterns with new ways to combine and apply
partially-defined, context-sensitive code using a fusion of functional and
object-oriented techniques.  For example, a (lazy) pair can be represented in
Scheme as procedure that takes \scm|`fst| or \scm|`snd| as an argument and
returns the respective element.  The two-by-two nested pair $((1, 2), (3, 4))$
can then be defined through multiple equations describing chains of projection
like so:
\begin{scheme}
(define*
  [((quad `fst) `fst) = 1]
  [((quad `fst) `snd) = 2]
  [((quad `snd) `fst) = 3]
  [((quad `snd) `snd) = 4])
\end{scheme}

Suppose we want to override the diagonal elements---the first of the first,
and the second of the second---with new values.  The function
\scm|(diag x y z)| below returns \scm|x| as the very first, \scm|y| as the very
last, and is the same as \scm|z| elsewhere
\begin{scheme}
(define*
  [(((diag x y z) `fst) `fst) = x]
  [(((diag x y z) `snd) `snd) = y]
  [  (diag x y z)             = z])
\end{scheme}
so that \scm|(diag 50 60 quad)| represents $((50, 2), (3, 60))$.  Notice the use
of the third ``fall-through'' case taken whenever the first two cases don't
apply.  Crucially, the context this fall-through case matches is \emph{less
  specified} than the earlier cases, so \scm|(((diag 50 60 quad)   `snd)  `fst)|
should simplify to \scm|((quad `snd)  `fst)  = 3|.  Effectively, the above
3-clause definition of \scm|diag| on two-by-two pairs is short-hand for the
following expansion that manually elaborates all four cases:
\begin{scheme}[basicstyle=\basicstyle\footnotesize]
(define*
  [(((diag x y z) `fst) `fst) = x]
  [(((diag x y z) `snd) `snd) = y]
  [(((diag x y z) `fst) `snd) = ((z `fst) `snd)]
  [(((diag x y z) `snd) `fst) = ((z `snd) `fst)])
\end{scheme}

As the notation suggests, we can understand this code through equational
reasoning---replacing call sites matching the left-hand side with the right-hand
side.  However, it is not so obvious how to convert this into an operational
semantics capable of directly calculating the result of any source-level
program, without elaborations like the above.  This becomes even more
challenging when taking into account other forms of higher-order composition
that can be performed at run-time: multiple clauses can be stitched together
vertically to handle undefined cases, and individual clauses can be extended
horizontally to consider more context or side conditions before committing to a
right-hand side.  The implementation \cite{CoScheme} gives a semantics based on
macro expansion into a small subset of Scheme (\ie $\lambda$-calculus), but only
describes the semantics via translation to the target language and not in terms
of the source-level language itself.

Thankfully, there is a general-purpose technique for converting one form of
semantics into another!  Danvy's functional correspondence
\cite{FunctionalCorrespondence} and syntactic correspondence
\cite{SyntacticCorrespondence} show how to use off-the-shelf program
transformations to derive semantic artifacts---operational semantics, abstract
machines, and continuation-passing style transformations---from one another in a
way that is correct by construction.  In other words, we can (mostly
mechanically) generate functional and correct operational semantics directly
from the macro definition implementing copatterns \cite{CoScheme}.  Yet, this is
not just an exercise of turning the crank.  Doing so reveals connections between
copattern matching with delimited control and a form of first-class handlers.

Our journey today is a round-trip hike in the semantic park.  Starting from a
simplified source copattern calculus with a straightforward operational
semantics, we derive its corresponding continuation-passing style (CPS)
transformation.  Then after making a few adjustments generalizing it to match
the macro definition found in \cite{CoScheme}, we turn around and walk backwards
to derive a direct-style operational semantics for compositional copattern
matching.  More specifically, our technical contributions are as follows:
\begin{itemize}
\item (\Cref{sec:derive-copat}) We derive a trio of semantic artifacts for a
  monolithic copattern calculus---small-step operational semantics, abstract
  machine, and continuation-passing style (CPS) translation---through a
  mechanical derivation based on functional code for copattern matching on
  contexts and a search to find the next redex.
\item (\Cref{sec:refactor}) We refactor the monolithic copattern calculus into
  more compositional primitives that give first-class control over
  \begin{enumerate*}[(1)]
  \item the (delimited) calling context and
  \item the options for how to handle match failure.
  \end{enumerate*}
  The CPS translation corresponds to the macro definition in \cite{CoScheme},
  and shows that the compositional calculus is a conservative extension of the
  monolithic one.
\item (\Cref{sec:control-copat}) Going in reverse, we walk back from the CPS
  translation of compositional copatterns and derive the missing semantic
  artifacts---an abstract machine and small-step operational semantics---that
  correspond by construction to the CPS translation.
\end{itemize}
The main highlights of these derivations are shown here, illustrating the
Haskell code that corresponds to the various semantic artifacts for copattern
matching.  To find the detailed step-by-step process, which serves as their
proof of correctness, see \url{https://github.com/pdownen/derive-copat}.

Before diving into the semantics, we begin with an example that illustrates the
added expressive power that compositionality gives to copatterns in functional
code, by infusing it with some techniques from the object-oriented paradigm.

\section{Expressiveness of Copatterns}
\label{sec:example}

To illustrate the expressiveness of copatterns, let's consider how they can be
used to write an interpreter for arithmetic expressions in Scheme.%
\footnote{The library defining the Scheme and Racket macros used in these
  examples can be found at \url{https://github.com/pdownen/CoScheme}.}%
A simple way to represent expressions is as trees where the nodes start with the
(symbolic) name of an operator followed by sub-expressions, and leaves are just
number literals.  For a simple language, the trees representing addition and
multiplication follow the patterns \scm|`(add ,l ,r)| and \scm|`(mul ,l ,r)|.%
\footnote{This is \emph{quasiquote} syntax, where the backquote \scm|`...| means
  to interpret the expression literally as a data structure---leaving names as
  quoted symbols---except for internal ``unquoting'' form \scm|,l| which means
  to insert the value of \scm|l| in place.  When used as a pattern, the form
  \scm|`...| matches against a nested tree data structure of that shape, and
  \scm|,l| inside of a backquote pattern is a wildcard variable that can match
  any value, which gets bound to \scm|l|.}
The arithmetic interpreter \scm|arith| can then be defined as an object
\cite{CoScheme} by (co)pattern matching on calls of the form
\scm|(arith `eval e)| where \scm|`eval| signals the request to evaluate an
expression tree \scm|e|:
\begin{scheme}
(define-object
  [(arith `eval n) (try-if (number? n))
   = n]
  [(arith `eval `(add ,l ,r))
   = (+ (arith `eval l) (arith `eval r))]
  [(arith `eval `(mul ,l ,r))
   = (* (arith `eval l) (arith `eval r))])

(define expr0 `(add 1 (mul 2 3)))
\end{scheme}
The first clause contains a \emph{guard} \scm|try-if| used to check that the
variable \scm|n| is bound to a numeric value; if not, the next two clauses are
tried, matching the two operators.  Following these equations derives
\scm|(arith `eval expr0)   = 7|.

Now, what if we wanted to add a new expression node, such as (unary) negation
\scm|`(neg ,e)|?  We can't just define a new function which handles the
\scm|`neg| operation and calls the old one for all other cases, like this
\emph{wrong} extension of \scm|arith|:
\begin{scheme}
(define-object
  [(arith-wrong `eval `(neg ,e))
   = (- (arith-wrong `eval e))]
  [(arith-wrong `eval e) = (arith `eval e)])

(define expr1 `(add 1 (neg (mul 2 3))))
\end{scheme}
Certain special cases will work.  Passing in an expression from the old syntax
will produce the correct answer---like \scm|(arith-wrong `eval expr0)    = 7|---or
limited uses of negation % Prevent --- from overflowing text box
---like \scm|(arith-wrong `eval `(neg 5))    = -5|.
However, examples where old and new operations are nested within one another
produce an error, like \scm|(arith-wrong `eval expr1)|.

In a functional language with built-in support for algebraic data types---like
Haskell and ML-family languages---we would have to copy-and-paste the old code
and add one extra clause for the new case.  Instead, with compositional
copatterns, we can reuse the old code as-is by composing it with another object
that defines the new behavior.  A correct way to extend the arithmetic evaluator
looks like:
\begin{scheme}
(define arith-ext
  (arith `compose
   (object
    [(self `eval `(neg ,e))
     = (- (self `eval e))])))
\end{scheme}
where the tag \scm|`compose| denotes an implicit inherited method of objects
that combines the clauses vertically in an either-or fashion.  That way, it
correctly calculates answers to both old expressions
\scm|(arith-ext `eval expr0)    = 7| and new ones
\scm|(arith-ext `eval expr1)    = -5|.  Textually, it is as if the underlying
runtime system did the copy-and-pasting for us, behaving exactly like the
equivalent expanded definition:%
\footnote{
  \renewcommand{\basicstylesize}{\footnotesize}%
  We renamed the recursive evaluator to \scm|arith-ext| in each equation
  for ease of reading, including the ones that were copied from the original
  \scm|arith| definition.  This is neither necessary nor changes the semantics,
  because \scm|define*| and \scm|define-object| are based on \emph{open
    recursion}, so that the variables used internally for recursive references
  on the right-hand sides of equations are different from the variable bound for
  external use by outside code.  In fact, the recursive name used to refer back
  to the object can be different between each line, and be different from the
  top-level name.}
\begin{scheme}[basicstyle=\basicstyle\footnotesize]
(define-object
  [(arith-ext `eval n) (try-if (number? n))
   = n]
  [(arith-ext `eval `(add ,l ,r))
   = (+ (arith-ext `eval l) (arith-ext `eval r))]
  [(arith-ext `eval `(mul ,l ,r))
   = (* (arith-ext `eval l) (arith-ext `eval r))]
  [(arith-ext `eval `(neg ,e))
   = (- (arith-ext `eval e))])
\end{scheme}
So we can add as many operators as we want to the evaluator without modifying
any old code. Fantastic!

However, what if we attempted a more radical change, such as extending
arithmetic to algebraic expressions with variables in them?  When evaluating a
variable, we need access to a mapping from variables to numbers to look up its
value.  Unfortunately, just composing clauses together vertically will no longer
suffice.  Typically, we would need to thread this environment around in every
other case where it's not needed, requiring careful surgery of the old code.

Instead, we can take a page out of the object-oriented book and think about
objects that ``encapsulate'' additional information within them.  For instance,
an object might hold onto exactly the variable-to-number environment we need.
Fortunately, we already have all the tools at our disposal to get the job
done---without needing to introduce the baggage of a whole class system.
Rather, we can model lightweight ``classes'' \ala JavaScript as functions
(representing the constructor) that return an \scm|object|.  For our needs, the
simplest class of objects containing a given variable-to-number environment is
defined by a single accessor method:
\begin{scheme}
(define (with-env dict)
  (object [(_ `env) = dict]))
\end{scheme}
So that \scm|((with-env e)  `env)  = e|.  This internal functionality can be
composed with our existing arithmetic evaluator and a new clause to \scm|`eval|
that says how to look up a variable (represented as a plain symbol) in an
expression tree.
\begin{scheme}
(define (alg dict)
  (arith-ext `compose
   (with-env dict)
   (object
    [(self `eval x) (try-if (symbol? x))
     = (dict-ref (self `env) x)])))

(define env-xy `((x . 10) (y . 20)))
(define expr2  `(add x (neg (mul 2 y))))
\end{scheme}
The algebraic evaluator gives the same answers on all the old examples without
having to use its environment:
\begin{scheme}
((alg env-xy) `eval expr0) = 7
((alg env-xy) `eval expr1) = 5
\end{scheme}
But now, it can handle new expressions that contain symbolic variables in them,
like
\begin{scheme}
((alg env-xy) `eval expr2) = -30
\end{scheme}
which is what we would get from performing the following manual textual revision and
inlining:
\begin{scheme}
(define (alg dict)
  (object
   [(self `eval n) (try-if (number? n))
    = n]
   [(self `eval `(add ,l ,r))
    = (+ (self `eval l) (self `eval r))]
   [(self `eval `(mul ,l ,r))
    = (* (self `eval l) (self `eval r))]
   [(self `eval `(neg ,e))
    = (- (self `eval e))]
   [(self `env)
    = dict]
   [(self `eval x) (try-if (symbol? x))
    = (dict-ref (self `env) x)]))
\end{scheme}
but all without touching any old code!

\section{Deriving Copatterns: A Journey of Small Steps to the Land of Continuations}
\label{sec:derive-copat}

In order to derive a semantics that can handle the types of open recursion and
composition of partial functions from \cref{sec:example}, let's consider a small
core calculus of copattern matching in \cref{fig:block-copat-spec}.  This
calculus has bound variables ($x$), applications of arguments (as $M~N$), and
projecting a specific index tagged $X$ (as $M~X$).  The only two other features
are:
\begin{itemize}
\item \emph{Open recursion:} To simplify the formalization of open-ended,
  late-bound self-reference, we follow a model close to the Python programming
  language.  The ``self'' pointer is given explicitly as the first argument, so
  calling a self-recursive function begins with an explicit punctuation (written
  $M.$) signaling that a copy of the same function should be passed first (as
  $M~M$).  The self-duplicating $\delta$ rule captures this step.
\item \emph{Copattern $\lambda$s:} Instead of just taking a single parameter,
  $\lambda$-abstractions are built out of a sequence of \emph{options} ($O$) for
  mapping \emph{copatterns} ($L$) on the left-hand side to a new term on the
  right-hand side.  Copatterns are a sequence of distinct parameter abstractions
  ($x~L$) and checks against specific index projections ($X~L$), until the empty
  copattern ($\varepsilon$) signals that no more information is needed to decide
  to return the right-hand side.  The resolution of a copattern $\lambda$ is a
  complex process, modeled by the $\beta$ rule, that checks each copattern in
  turn against the $\lambda$'s context $C$.  If somehow the $\lambda$ appears in
  a context $C$, and one of its copatterns $L_i$ can match $C$ by substituting
  for its bound variables, then its right-hand side $M_i$ should be returned
  under the same substitution, as long as $L_i$ is the first such match.
\end{itemize}

\begin{figure}
\centering
\begin{alignat*}{2}
  \mathit{Variable} &\ni{}& x, y, z
  &::= \dots
  \\
  \mathit{Index} &\ni{}& X, Y, Z
  &::= \dots
  \\
  \mathit{Term} &\ni{}& M, N
  &::= x
  \mid M ~ N
  \mid M ~ X
  \mid M .
  \mid \lambda \{\many{O}\}
  \\
  \mathit{Option} &\ni{}& O
  &::= L \to M
  \\
  \mathit{Copat} &\ni{}& L
  &::= \varepsilon
  \mid x ~ L
  \mid X ~ L
\end{alignat*}
\begin{align*}
  \begin{aligned}
    (\delta)&&
    M. &= M ~ M
    \\
    (\beta)&&
    C[\lambda\{\many[1 \leq i \leq n]{L_i \to M_i}\}]
    &=
    M_j\vect{\subst{x}{N}}
  \end{aligned}
  \\
  \left(
    \begin{alignedat}{2}
      &\text{if } &C &= L_j\vect{\subst{x}{N}}
      \\
      &\text{and } &\forall i < j,~
      \not\exists \vect{N},~ C &= L_i\vect{\subst{x}{N}}
    \end{alignedat}
  \right)
\end{align*}
\caption{Equational specification of monolithic copatterns.}
\label{fig:block-copat-spec}
\end{figure}

\begin{example}
\label{ex:count}
  
To see how this core calculus models recursion and matching, consider this
infinite stream object (with a $\mathit{Head}$ and $\mathit{Tail}$) that counts
$\mathit{From}$ an initial number:
\begin{align*}
  % \mathit{nats} &= \mathit{count}.From~0
  % \\
  \mathit{count} &=
  \lambda
  \begin{aligned}[t]
    \{&
    \mathit{self}~\mathit{From}~x~\mathit{Head} \to x
    \\
    \mid&
    \mathit{self}~\mathit{From}~x~\mathit{Tail} \to \mathit{self}.\mathit{From}(\mathit{succ}~x)
    \}
  \end{aligned}
\end{align*}
Starting the $\mathit{count}$ from 0 and accessing the third element (via two
$\mathit{Tail}$s and a $\mathit{Head}$) shows how $\delta$ and $\beta$ enable
recursion:
\begin{align*}
  % &
  % \underline{\mathit{nats}}~\mathit{Tail}~\mathit{Tail}~\mathit{Head}
  % \\
  % &=
  &
  \underline{\mathit{count}.}\mathit{From}~0~\mathit{Tail}~\mathit{Tail}~\mathit{Head}
  \\
  &=
  \underline{
    \mathit{count}~\mathit{count}~\mathit{From}~0~\mathit{Tail}
  }
  ~\mathit{Tail}~\mathit{Head}
  &(\delta)
  \\
  &=
  \underline{\mathit{count}.}\mathit{From}(\mathit{succ}~0)~\mathit{Tail}~\mathit{Head}
  &(\beta)
  \\
  &=
  \underline{
    \mathit{count}~\mathit{count}~\mathit{From}~(\mathit{succ}~0)~\mathit{Tail}
  }
  ~\mathit{Head}
  &(\delta)
  \\
  &=
  \underline{\mathit{count}.}\mathit{From}(\mathit{succ}~(\mathit{succ}~0))~\mathit{Head}
  &(\beta)
  \\
  &=
  \underline{
    \mathit{count}~\mathit{count}
    ~\mathit{From}~(\mathit{succ}~(\mathit{succ}~0))
    ~\mathit{Head}
  }
  &(\delta)
  \\
  &=
  (\mathit{succ}~(\mathit{succ}~0))
  &(\beta)
\end{align*}
\end{example}

\subsection{Small-step operational semantics}

\begin{figure}
\centering
\begin{haskell}
data Term i a
  = Var a
  | Term i a :*: Term i a
  | Term i a :@: i
  | Dot (Term i a)
  | Obj [Option i a]

data Option i a = Copattern i a :-> Term i a

data Copattern i a
  = Nop
  | a :* Copattern i a
  | i :@ Copattern i a

instance Semigroup (Copattern i a) where
  Nop      <> q' = q'
  (x :* q) <> q' = x :* (q <> q')
  (i :@ q) <> q' = i :@ (q <> q')
instance Monoid (Copattern i a) where
  mempty = Nop

type Question i a = Copattern i (Term i a)

ask :: Term i a -> Question i a -> Term i a
ask m Nop      = m
ask m (n :* q) = ask (m :*: n) q
ask m (i :@ q) = ask (m :@: i) q

type Env a b = [(a, b)]
type TermEnv i a = Env a (Term i a)
(//)::Eq a=>Term i a->TermEnv i a->Term i a
\end{haskell}
\caption{Syntax trees as an algebraic data type.}
\label{fig:block-syntax-code}
\end{figure}

\begin{figure}
\centering
\renewcommand{\basicstylesize}{\footnotesize}
\begin{haskell}
data Redex    i a
  = Introspect (Term i a)
  | Respond [Option i a]
  | FreeVar a
data Reduct   i a
  = Reduced (Term i a)
  | Unhandled
  | Unknown a
data Followup i a
  = Next (Reduct i a) (Question i a)
  | More (Copattern i a) (Term i a)
         [Option i a] (Question i a)

reduce :: (Eq i, Eq a)
       => Redex i a -> Question i a
       -> Followup i a
reduce (Introspect m) q
  = Next (Reduced $ m :*: m) q
reduce (FreeVar x) q
  = Next (Unknown x) q
reduce (Respond (lhs :-> rhs : ops)) q
  = case suffix match of
      Followup q'  -> Next (Reduced rhs') q'
      Unasked lhs' -> More lhs' rhs' ops q
      Mismatch _ _ -> reduce (Respond ops) q
  where match = comatch lhs q
        rhs'  = rhs // prefix match
reduce (Respond []) q
  = Next Unhandled q

data CoMatch   i a b
  = Comatch { prefix :: Env a b,
              suffix :: Remainder i a b }

data Remainder i a b
  = Followup (Copattern i b)
  | Unasked  (Copattern i a)
  | Mismatch (Copattern i a) (Copattern i b)

comatch :: Eq i
        => Copattern i a -> Copattern i b
        -> CoMatch i a b
comatch Nop        cxt
  = Comatch [] (Followup cxt)
comatch lhs        Nop
  = Comatch [] (Unasked lhs)
comatch (x :* lhs) (y :* cxt)
  = Comatch ((x, y) : prefix q) (suffix q)
  where q = comatch lhs cxt
comatch (i :@ lhs) (j :@ cxt)
  | i == j = comatch lhs cxt
comatch lhs        cxt
  = Comatch [] (Mismatch lhs cxt)
\end{haskell}
\caption{An implementation of small-step reduction.}
\label{fig:block-reduce-code}
\end{figure}

\begin{figure}
\centering
\begin{haskell}
data Found i a
  = Asked (Redex i a) (Question i a)

search :: Term i a -> Found i a
search (Var x)   = Asked (FreeVar x) Nop
search (Dot m)   = Asked (Introspect m) Nop
search (Obj ops) = Asked (Respond ops) Nop
search (m :*: n) = case search m of
  Asked r q -> Asked r $ q <> n :* Nop
search (m :@: i) = case search m of
  Asked r q -> Asked r $ q <> i :@ Nop
\end{haskell}
\caption{Searching for the next redex.}
\label{fig:block-search-code}
\end{figure}

\begin{figure}
\centering
\renewcommand{\basicstylesize}{\footnotesize}
\begin{haskell}
data Decomp i a = Asked (Redex i a) (Question i a)

recomp :: Term i a -> Question i a -> Term i a
recomp = ask

decomp :: Term i a -> Decomp i a
decomp m = refocus m Nop

refocus :: Term i a -> Question i a -> Decomp i a
refocus (Var x)   k = Asked (FreeVar x) k
refocus (Dot m)   k = Asked (Introspect m) k
refocus (Obj eqs) k = Asked (Respond eqs) k
refocus (m :*: n) k = refocus m $ n :* k
refocus (m :@: i) k = refocus m $ i :@ k
\end{haskell}
\caption{Decomposing a term into a redex and question.}
\label{fig:block-decomp-code}
\end{figure}

\begin{figure}
\centering
\renewcommand{\basicstylesize}{\footnotesize}
\begin{haskell}
data Answer i a
  = Under (Copattern i a) (Term i a)
          [Option i a]    (Question i a)
  | Raise (Question i a)
  | Stuck a (Question i a)

eval::(Eq a, Eq i) => Term i a -> Answer i a
eval m = iter $ decomp m

iter::(Eq a, Eq i) => Decomp i a -> Answer i a
iter (Asked r q) = case reduce r q of
  Next (Reduced m) k -> eval $ recomp m k
  Next (Unknown x) k -> Stuck x k
  Next Unhandled   k -> Raise k
  More lhs rhs eqs k -> Under lhs rhs eqs k
\end{haskell}
\caption{Functional small-step interpreter loop.}
\label{fig:block-small-step-code}
\end{figure}

The equational axioms $\beta$ and $\delta$ can specify \emph{why} an answer is
correct, but they don't give an algorithmic method showing \emph{how} to get
there.  So let's write an algorithm!

The first step is to represent syntax trees as a concrete data structure, which
we can do in Haskell as shown in \cref{fig:block-syntax-code}.  We use infix
constructors \hs|m:*:n| and \hs|m:@:x| for function application ($M~N$) and
indexing ($M~X$), respectively, and the other constructors are for variables
(\hs|Var "x"| as $x$) the ``dot'' operator (\hs|Dot m| as $M.$) and copattern
$\lambda$-objects (\hs|Obj [l :-> m, ...]| as $\lambda\{L \to M \mid \dots\}$).
Copatterns are built using similar infix constructors and end with no-op
\hs|Nop|, forming a stylized list that we can concatinate using \hs|<>|.  Of
note, we abstract variable identifiers and projection indexes as generic types
\hs|a| and \hs|i|, respectively, which will come in handy several times.  A
special case of \hs|Copattern|s are \hs|Question|s---contexts that might match
copatterns---given by filling the variables (left of \hs|:*|) in a copattern
with a \hs|Term i a|.  Being contexts, we can plug a term into a question via
\hs|ask|.

Second, we have to implement a single step of reduction, which is shown in
\cref{fig:block-reduce-code}.  The \hs|reduce| function takes a (potential)
\hs|Redex| \emph{and} a \hs|Question| to produce some \hs|Followup| result:
either a \hs|Reduct| is asked the next \hs|Question|, or comatching needed more
context.  The \hs|comatch| function compares the left-hand-side against a
question, creating a substitution environment and saying if there is a full
match (producing a followup question out of the remaining context), an
incomplete match (producing the copattern that continues past the given
question), or a mismatch.  In the case of a full match, \hs|reduce| will
substitute (via \hs|//|) the matching environment into the right-hand side and
return the followup question.  In the case of an incomplete match, reduction is
blocked, and in the case of a mismatch, the next option is tried.

Third, we must spell out how to find the next redex in a term.  A direct-style
\hs|search| function is shown in \cref{fig:block-search-code}, which identifies
both a redex as well as the question asked of it.  Following the syntactic
correspondence methodology \cite{SyntacticCorrespondence}, we can derive a
tail-recursive decomposition function from \hs|search| using standard program
transformations: CPS transformation, defunctionalization
\cite{DefinitionalInterpreters}, and compressing corridor transitions (\ie
inlining and simplifying known function calls to partially-known arguments).
Along the way, it becomes clear that the evaluation contexts are isomorphic to
the \hs|Question|s we need for reduction, which comes from the choice of
call-by-name evaluation order.  Swapping to the existing representation and
applying the monoid laws gives \hs|decomp| in \cref{fig:block-decomp-code}.

Finally, the small-step interpreter is given in
\cref{fig:block-small-step-code}, which repeatedly decomposes a term, reduces
it, recomposes the result, and loops.  This algorithm is equivalent to the
following relational small-step semantics, presented in terms of \emph{inside
  out} evaluation contexts and a $\mathit{comatch}$ function:
\begin{alignat*}{2}
  \mathit{Cont} &\ni{}& K
  &::= \varepsilon
  \mid N ~ K
  \mid X ~ K
\end{alignat*}
\begin{gather*}
  \begin{aligned}
    (\delta)&&
    K[M.] &\mapsto K[M ~ M]
    \\
    (\beta)&&
    K[\lambda\{\many[1 \leq i \leq n]{L_i \to M_i}\}]
    &\mapsto
    K'[M_j\many{\subst{x}{N}}]
  \end{aligned}
  \\
  \qquad\qquad
  \left(
    \begin{alignedat}{2}
      &\text{if } &\mathit{comatch}(L_j, K) &= (\vect{\subst{x}{N}}, K')
      \\
      &\text{and } &\forall i < j,~ \mathit{comatch}(L_i, K) &= \mathit{Mismatch}
    \end{alignedat}
  \right)
\end{gather*}

% \begin{alignat*}{2}
%   \mathit{EvalCxt} &\ni{}& E
%   &::= \hole
%   \mid E ~ N
%   \mid E ~ X
%   \\
%   \mathit{Cont} &\ni{}& K
%   &::= \varepsilon
%   \mid N ~ K
%   \mid X ~ K
% \end{alignat*}

% \begin{align*}
%   (E~N)[M] &= E[M] ~ N
%   \\
%   (E~X)[M] &= E[M] ~ X
%   \\
%   \hole[M] &= M
% \end{align*}

% \begin{align*}
%   (N~K)[M] &= K[M~N]
%   \\
%   (X~K)[M] &= K[M~X]
%   \\
%   \varepsilon[M] &= M
% \end{align*}

% \begin{align*}
%   \insideout{\hole} &= \varepsilon
%   &
%   \outsidein{\varepsilon} &= \hole
%   \\
%   \insideout{E~N} &= \insideout{E} \comp (N~\varepsilon)
%   &
%   \outsidein{N~K} &= \outsidein{K} \comp (\hole~N)
%   \\
%   \insideout{E~X} &= \insideout{E} \comp (X~\varepsilon)
%   &
%   \outsidein{X~K} &= \outsidein{K} \comp (\hole~X)
% \end{align*}

\subsection{Abstract machine}

\begin{figure}
\centering
\renewcommand{\basicstylesize}{\footnotesize}
\begin{haskell}
eval m = refocus m Nop

refocus (Var x)   k = Stuck x k
refocus (Dot m)   k = refocus m $ m :* k
refocus (Obj os)  k = case os of
  []               -> Raise k
  lhs :-> rhs : os -> comatch lhs k rhs os k
refocus (m :*: n) k = refocus m $ n :* k
refocus (m :@: i) k = refocus m $ i :@ k

comatch Nop         cxt       rhs _  _
  = refocus rhs cxt
comatch lhs         Nop       rhs os q
  = Under lhs rhs os q
comatch (x :* lhs) (y :* cxt) rhs os q
  = comatch lhs cxt (rhs // [(x,y)]) os q
comatch (i :@ lhs) (j :@ cxt) rhs os q
  | i == j
  = comatch lhs cxt rhs os q
comatch lhs         cxt       _   os q
  = refocus (Obj os) q
\end{haskell}
\caption{Tail-recursive abstract machine interpreter.}
\label{fig:block-machine-code}
\end{figure}

Continuing on, we can use the syntactic correspondence to transform the
direct-style small-step interpreter into a tail-recursive abstract machine.
First, we short-cut the recompose-decompose step and instead continue by
refocusing in place.
\begin{lemma}
  \label{thm:block-refocusing}

  \hs|decomp (recomp m k) = refocus m k|.
\end{lemma}
\begin{proof}
  By induction on \hs|k :: Question i a|.
\end{proof}
From there, it is a matter of applying standard program transformations: loop
fusion, compressing corridor transitions, and deforesting intermediate data
structures.  In order to fuse \hs|reduce| and \hs|comatch| in with the
small-step interpreter, we need to give reduction the same treatment as
\hs|search| to put it into a tail-recursive form: CPS transforming,
defunctionalizing \cite{DefinitionalInterpreters}, and compressing corridor
transitions.  As a simplification, the contexts produced by defunctionalization
are isomorphic to substitution environments, whose order is irrelevant (assuming
distinct names).  We also modify copattern matching to substitute immediately
when available, since substitution reassociates.
\begin{lemma}
  \label{thm:block-subst-reassoc}
  \hs|m // (env ++ env') = (m // env) // env'|
\end{lemma}
\begin{proof}
  By induction on \hs|m :: Term i a|.
\end{proof}

The end result is the abstract machine interpreter shown in
\cref{fig:block-machine-code}.  We can rephrase this code into a traditional
stepping relation using machine states of forms
\begin{enumerate*}[(1)]
\item a refocusing or reduction state $\cut{M}{K}$, or
\item a copattern-matching state
  $\braket{L \cmid K \cmid M \cmid \many{O} \cmid K}$.
\end{enumerate*}
The initial state for evaluating $M$ is $\cut{M}{\varepsilon}$, and the final
states are one of:
\begin{enumerate*}[(1)]
\item stuck on an unknown variable $\cut{x}{K}$,
\item an unhandled question $\cut{\lambda\{\}}{K}$, or
\item an underspecified question
  $\braket{L \cmid \varepsilon \cmid M \cmid \many{O} \cmid K}$ where $L \neq \varepsilon$.
\end{enumerate*}
This gives us the following stepping relation:
\begin{itemize}
\item Refocusing steps:
\begin{align*}
  \cut{M~N}{K}
  &\mapsto
  \cut{M}{N ~ K}
  &
  \cut{M~X}{K}
  &\mapsto
  \cut{M}{X ~ K}
\end{align*}

\item Reduction steps:
  \begin{align*}
  \cut{M.}{K}
  &\mapsto
  \cut{M}{M ~ K}
  \\
  \cut{\lambda\{L \to M \mid \many{O}\}}{K}
  &\mapsto
  \braket{L \cmid K \cmid M \cmid \many{O} \cmid K}
  % \\
  % \cut{\lambda\{\}}{K}
  % &\not\mapsto
  % \\
  % \cut{x}{K}
  % &\not\mapsto
\end{align*}

\item Copattern-matching steps:
\begin{align*}
  \braket{x ~ L \cmid N ~ K' \cmid M \cmid \many{O} \cmid K}
  &\mapsto
  \braket{L \cmid K' \cmid M\subst{x}{N} \cmid \many{O} \cmid K}
  \\
  \braket{X ~ L \cmid X ~ K' \cmid M \cmid \many{O} \cmid K}
  &\mapsto
  \braket{L \cmid K' \cmid M \cmid \many{O} \cmid K}
  \\
  \braket{\varepsilon \cmid K' \cmid M \cmid \many{O} \cmid K}
  &\mapsto
  \cut{M}{K'}
  \\
  \braket{L \cmid \varepsilon \cmid M \cmid \many{O} \cmid K}
  &\not\mapsto
  \qquad\qquad\qquad\quad\;\;
  (\text{if } L \neq \varepsilon)
  \\
  \braket{L \cmid K' \cmid M \cmid \many{O} \cmid K}
  &\mapsto
  \cut{\lambda\{\many{O}\}}{K}
  \quad(\text{otherwise})
\end{align*}
\end{itemize}

\begin{remark}
  \label{rm:environment-machine}
  Note that this abstract machine inefficiently traverses terms many times to
  perform substitution.  To derive an environment-based machine that more
  efficiently uses closures and lookup, see appendix \cref{sec:environment-machine}.
\end{remark}

\subsection{Continuation-passing style transformation}

\begin{figure}
\centering
\renewcommand{\basicstylesize}{\footnotesize}
\begin{haskell}
data Answer i a
  = Under (CPSTerm i a)
  | Raise (CPSQuestion i a)
  | Stuck a (CPSQuestion i a)

type CPSQuestion i a = Copattern i (CPSArg i a)
type CPSTerm i a = CPSQuestion i a -> Answer i a
type CPSOption i a = CPSTerm i a -> CPSTerm i a
type CPSCopat i a = CPSQuestion i a -> CPSOption i a

newtype CPSArg i a = Arg { useArg :: CPSTerm i a }

data CPSVar i a = Name a | Subs (CPSTerm i a)

instance Eq a => Eq (CPSVar i a) where
  Name x == Name y = x == y
  _      == _      = False

eval :: (Eq i, Eq a) => Term i a -> Answer i a
eval m = (term (fmap Name m)) Nop

term :: (Eq i, Eq a) => Term i (CPSVar i a)
     -> CPSTerm i a
term (Var (Name x)) = Stuck x
term (Var (Subs m)) = m
term (Dot m)   = \k -> (term m) (Arg(term m) :* k)
term (Obj os)  = options os
term (m :*: n) = \k -> (term m) (Arg(term n) :* k)
term (m :@: i) = \k -> (term m) (i :@ k)

options :: (Eq i, Eq a) => [Option i (CPSVar i a)]
        -> CPSTerm i a
options [] = Raise
options (lhs :-> rhs : os)
  = \q -> (comatch lhs rhs) q (options os) q

comatch :: (Eq i, Eq a) => Copattern i (CPSVar i a)
        -> Term i (CPSVar i a) -> CPSCopat i a
comatch Nop        rhs = \_ _   -> (term rhs)
comatch (x :* lhs) rhs = \q os -> \case
  (y :* k) -> (comatch lhs (rhs // [(x, n)])) q os k
    where n = Var (Subs (useArg y))
  Nop      -> Under $ (comatch (x :* lhs) rhs) q os
  _        -> os q
comatch (i :@ lhs) rhs = \q os -> \case
  (j :@ k)
  | i == j -> (comatch lhs rhs) q os k
  Nop      -> Under $ (comatch (i :@ lhs) rhs) q os
  _        -> os q
\end{haskell}
\caption{Continuation-passing style translation to Haskell.}
\label{fig:block-cps-code}
\end{figure}

At the end of our journey, we arrive at a continuation-passing style translation
to native Haskell functions, as shown in \cref{fig:block-cps-code}.  This
translation is derived from \cref{fig:block-machine-code} by
\begin{enumerate*}[(1)]
\item desugaring nested patterns into flat patterns on a single value,
\item $\eta$-reduction, and
\item applying the transition functions as soon as possible.
\end{enumerate*}
Of note, the code elaborates a detail that is usually suppressed in CPS
transformations: when going under a binder, we have to replace a \emph{syntactic
  name} from the source program with a \emph{semantic denotation} within the
CPS.  We accommodate this step in the middle of the CPS process by representing
variables as either a plain \hs|Name| or one that was \hs|Subs|tituted by a
\hs|CPSTerm|, and should be translated as-is.

To better understand the code, we can elide some of these details in the
corresponding CPS translation of terms $\den{M}$, lists of options
$\den{\many{O}}$, and individual copattern-matching options $\den{L \to M}$ in a
more traditional notation (where $\Rec$ denotes a recursive fixed point to
handle the non-structural recursion for under-application in \hs|comatch|):

\begin{itemize}
\item  Translating terms $\den{M}$:
\begin{gather*}
\begin{aligned}
  \den{x} &= x
  &
  \den{M.} &= \lambda k.~ \den{M} ~ (\den{M}, k)
  \\
  \den{M~N} &= \lambda k.~ \den{M} ~ (\den{N}, k)
  &
  \den{M~X} &= \lambda k.~ \den{M} ~ (X ~ k)
\end{aligned}
  \\
  \den{\lambda\{\many{O}\}} = \den{\many{O}}
\end{gather*}
\item  Translating lists of options $\den{\many{O}}$:
\begin{align*}
  \den{\varepsilon} &= \lambda k.~ k
  \\
  \den{L \to M \mid \many{O}}
  &=
  \lambda k.~ \den{L \to M} ~ k ~ \den{\many{O}} ~ k
\end{align*}
\item Translating one copattern-matching option $\den{L \to M}$:
\begin{align*}
  \den{\varepsilon \to N} &= \lambda q. \lambda f.~ \den{N}
  \\
  \den{x ~ L \to N} &= \Rec r = \lambda q. \lambda f. \lambda k.
  \\[-1ex]
  &\qquad
  \Case k \Of
  \begin{alignedat}[t]{2}
    &(x, k') &&\to \den{L \to N} ~ q ~ f ~ k'
    \\[-1ex]
    &() &&\to r ~ q ~ f
    \\[-1ex]
    &k &&\to f ~ q
  \end{alignedat}
  \\
  \den{X ~ L \to N} &= \Rec r = \lambda q. \lambda f. \lambda k.
  \\[-1ex]
  &\qquad
  \Case k \Of
  \begin{alignedat}[t]{2}
    &(X ~ k') &&\to \den{L \to N} ~ q ~ f ~ k'
    \\[-1ex]
    &() &&\to r ~ q ~ f
    \\[-1ex]
    &k &&\to f ~ q
  \end{alignedat}
\end{align*}
\end{itemize}

\begin{remark}
  Typically, we would keep going and refunctionalize continuations---in this
  case, \hs|Question|s---into first-class functions.  This step becomes
  difficult in cases like ours, compounded next in \cref{sec:refactor}, which is
  not in the image of ordinary defunctionalization.  Refunctionalization can be
  made total by using copattern matching on
  codata~\cite{RefunctionalizationCopatterns}---but that verges on begging the
  question by defining copatterns in terms of copatterns.  But not to worry!
  The CPS transformation given here corresponds to a well-known and well-behaved
  one for call-by-name $\lambda$-calculus, based on a concrete representation of
  continuations as pair and sum types
  \cite{ContinuationModels,ClassicalLogicContinuationSemantics,CBNSyntacticCPS,AnswerTypePolyEta},
  so we can stop here.
\end{remark}

Because all three semantics have been derived from a single origin using correct
program transformations, we now get a theorem about their correspondence that is
correct by construction.  For simplicity, we single out raising an unanswered
question as the main observation of programs.
\begin{theorem}
\label{thm:block-semantic-correspondence}

The three \hs|eval| functions are equal, \ie the following relations between $M$
and $K$ are all equivalent:
\begin{enumerate}[(a)]
\item $M \mapsto^* K[\lambda\{\many{L_i \to M_i}\}]$ such that, for all $i$, \\
  $\mathit{comatch}(L_i,K) = \mathit{Mismatch}$.
\item $\cut{M}{\varepsilon} \mapsto^* \cut{\lambda\{\}}{K}$.
\item $\den{M}() \mapsto^* \den{K}$, where $\den{K}$ is
  \begin{align*}
    \den{\varepsilon} &= ()
    &
    \den{N~K} &= (\den{N}, \den{K})
    &
    \den{X~K} &= X~\den{K}
  \end{align*}
\end{enumerate}
\end{theorem}

\section{Refactoring Syntax and Semantics: A Short Rest Among the Lambdas}
\label{sec:refactor}

The joy of working with a CPS transformation, like the one we now have, is that
we can employ the high-powered theory of the $\lambda$-calculus to reason about
the semantics of our programming language.  The $\lambda$-calculus has a
thoroughly developed suite of semantic tools to help us prove properties of the
transformation, and the denotational style of CPS unlocks many out-of-order
simplifications and rewrites for free.  Let's now use this ability to refactor
our language.

\subsection{First refactor: Delimiting the context}

% Note, we already have a term to express ``first-class'' continuations.  Suppose
% we translate continuations $\den{K}$ (\ie copatterns with terms substituted for
% the bound variables) directly, where $k_0$ stands for the ``initial continuation,'' as:
% \begin{align*}
%   \den{\varepsilon} ~ k_0 &= k_0
%   \\
%   \den{N ~ K}~k_0 &= (\den{N}, (\den{K}~k_0))
%   \\
%   \den{X ~ K}~k_0 &= X ~ (\den{K}~k_0)
% \end{align*}
% \begin{lemma}
%   \label{thm:cps-cont-refocus}
%   $\den{K[M]} ~ \den{K'} \mapsto^* \den{M} ~ (\den{K \compq K'})$
% \end{lemma}
% \begin{proof}
%   By induction on the syntax of $K$.
% \end{proof}
% \begin{corollary}
%   \label{thm:cps-cont-refocus-top}
%   $\den{K[M]}() \mapsto^* \den{M} ~ \den{K}$
% \end{corollary}

% Then notice how an empty $\lambda$, which responds to no continuation, will
% immediately raise it to the ``top-level'':
% \begin{align*}
%   \den{K[\lambda\{\}]}()
%   &=
%   \den{\lambda\{\}} ~ \den{K}
%   &(\cref{thm:cps-cont-refocus-top})
%   \\
%   &=
%   (\lambda k.~ k) ~ \den{K}
%   &(\text{def})
%   \\
%   &=
%   \den{K}
%   &(\beta)
% \end{align*}

% \begin{align*}
%   \den{\varepsilon \ans q \to M}
%   &=
%   \lambda q. \den{M}
% \end{align*}

One of the awkward aspects of the semantics so far has to do with unresolved
matching, when a copattern is expecting more information than the context
provides.  Currently, there is no way to tell when a question is really
``done,'' or if we are missing some part of the context.  As a consequence, the
CPS transformation has to handle cases of an empty stack by trying again to
consume more continuation.  This leads to the complex recursive structure of
copatterns (seen in the $\Rec$ fixed point) and confusion between different
types of continuations (\eg $x~L$ may get a pair or an empty unit).

To resolve these ambiguities, let's add a delimiter to the language, $M!$, that
signals the definite end to a question.  If copattern matching reaches the final
punctuation ($!$), then it knows there will never be more context arriving and
it can move on to the next option.  But if $M!$ signals the end of questioning,
we can never interrogate the answer with another question, so where else can the
delimiter appear except at the ``top'' of the whole program?

Giving the program some internal control over delimited questions amounts to a
form of call-by-name delimited control.  That is to say, a term can abstract
over a given question by giving it a name $q$ in $\ans q \to R$, where $R$ is a
\emph{response} that explains how to continue, which could be asking the same
question $M \ans q$, a different one $M \ans q'$, or a now explicitly empty
$M \ans \varepsilon$.  Instead of asking, a response could also raise an
unanswered question itself, to be handled at a higher level.

Extending the syntax we have so far, first-class delimited questions have the
following grammatical structure:
\begin{alignat*}{2}
  \mathit{Response} &\ni{}& R
  &::= q
  \mid \varepsilon
  \mid M \ans R
  \\
  \mathit{Term} &\ni{}& M, N
  &::= \dots
  \mid \ans q \to R
\end{alignat*}
Whose semantics is given by extending the CPS transform:
\begin{align*}
  \den{M \ans R}
  &=
  \den{M} ~ \den{R}
  &
  \den{\varepsilon}
  &=
  ()
  &
  \den{q}
  &=
  q
  \\[1ex]
  \den{\ans q \to R}
  &=
  \lambda q.~ \den{R}
\end{align*}
Notice how call-by-name delimited questions can compose, but the compositional
structure is opposite to call-by-value delimited control like shift and reset
\cite{AbstractingControl,FunctionalAbstractionTypedContexts}.  Rather than
composing multiple continuations like functions from inputs to outputs, we can
instead compose several terms---one after another---as functions from questions
to answers that handle the unanswered questions raised by the next one in line:
\begin{align*}
  \den{M \ans (N \ans R)}
  &=
  \den{M} ~ (\den{N} ~ \den{R})
  &
  \den{M \ans \varepsilon}
  &=
  \den{M} ~ ()
\end{align*}
The term that immediately raises its given question can be expressed
$\ans q \to q$.  Later in \cref{sec:control-copat}, it will be useful to have a
special form $\Raise$ to denote when this has happened, without having to take
any more steps to calculate the response:
\begin{align*}
  \den{\Raise}
  &=
  \lambda q. q
  = \den{\ans q \to q}
\end{align*}

% Note, we can wrap up functional forms of
% \begin{itemize}
% \item a jump/throw operation $\ans\_ \to (x \ans k)$
% \item an abort operation $\ans\_ \to (x \ans \varepsilon)$
% \item first-class abortive continuations $\lambda \{x \to \ans\_ \to (x \ans R)\}$
% % \item first-class option $\lambda \{x \to (O \ask x)\}$
% \end{itemize}

% This lets us express abortive control operators like \texttt{call/cc} and $\mathcal{C}$:
% \begin{align*}
%   \mathtt{call/cc}
%   &=
%   \lambda \{h \to \ans k \to h ~ \lambda\{x \to \ans\_ \to x \ans k\} \ans k\}
%   \\
%   \mathcal{C}
%   &=
%   \lambda \{h \to \ans k \to h ~ \lambda\{x \to \ans\_ \to x \ans k\} \ans {}\}
% \end{align*}

\subsection{Second refactor: Nesting copatterns}

Another source of difficulty is the complex structure of copattern-matching
$\lambda$-abstractions, which forces a monolithic matching algorithm.  In the
interest of factoring out individual aspects of copattern matching, we will
decompose the copattern options into smaller parts.  The first step is to
reassociate copatterns to the right, taking them one step at a time in the style
of currying, for example replacing $(x~L) \to M$ with $x \to (L \to M)$.  The
second step is to get a handle on how a match failure should proceed, since
every copattern option needs to do something if it can't respond to its context.
We write $O \ask M$ to mean that $O$ is the first option to answer the context
and, if it fails, the term proceeds as $M$ in the \emph{same} context.  Dually,
the option needs a way to signal success when the right-hand side is reached,
which we write $\ask x \to N$ to mean that $N$ is returned to this context, and
the (now untaken) failure alternative is bound to $x$ in case the program needs
to invoke it manually.

Nested options have the following grammatical structure:
\begin{alignat*}{2}
  \mathit{Term} &\ni{}& M, N
  &::= \dots
  \mid O \ask M
  \\
  \mathit{Option} &\ni{}& O
  &::= x \to O
  \mid X \to O
  \mid \ask x \to M
\end{alignat*}
On the one hand, we can understand the new forms (besides $\ask x \to M$) in
terms of the old ones:
\begin{align*}
  O \ask M &= \lambda\{O \mid \varepsilon \to M\}
  \\
  x \to (L \to M) &= (x~L) \to M
  \\
  X \to (L \to M) &= (X~L) \to M
\end{align*}
% From which we can derive the following CPS translation:
% \begin{align*}
%   \den{O \ask M}
%   &=
%   \lambda k.~ \den{O} ~ k ~ \den{M} ~ k
%   \\[1ex]
%   \den{x \to O}
%   &=
%   \lambda q. \lambda f. \lambda k.
%   \Case k \Of
%   \begin{alignedat}[t]{2}
%     &(x, k') &&\to \den{O} ~ q ~ f ~ k' \\
%     &k &&\to f ~ q
%   \end{alignedat}
%   \\
%   \den{X \to O}
%   &=
%   \lambda q. \lambda f. \lambda k.
%   \Case k \Of
%   \begin{alignedat}[t]{2}
%     &(X ~ k') &&\to \den{O} ~ q ~ f ~ k' \\
%     &k &&\to f ~ q
%   \end{alignedat}
%   \\
%   \den{\ask x \to M}
%   &=
%   \lambda \_. \lambda x. \den{M}
% \end{align*}
On the other hand, we can decompose the monolithic syntax into smaller, nested
pieces:
\begin{align*}
  % \lambda\{O\} &= O \ask \Raise
  % \\
  % O_1 \mid \dots \mid O_n &= \ask x \to ~ O_1 \ask (\dots \ask (O_n \ask x))
  % \\
  \lambda\{O_1 \mid \dots \mid O_n\}
  &=
  O_1 \ask (\dots \ask (O_n \ask \Raise)
  \\
  \varepsilon \to M &= \ask \_ \to M
  \\
  (x~L) \to O &= x \to (L \to O)
  \\
  (X~L) \to O &= X \to (L \to O)
\end{align*}
% Some new forms that are now expressible in the nested syntax:
% \begin{align*}
%   (L \ask x \to M) &= L \to (\ask x \to M)
%   \\
%   (L \ans q \to R) &= L \to (\ask \_ \to (\ans q \to R))
%   \\
%   (L \ans q \ask x \to R) &= L \to (\ask x \to (\ans q \to R))
% \end{align*}

More interestingly, we now have enough flexibility over options and their
failure modes to encode dynamic composition of copattern matching, as used by
the arithmetic evaluator in \cref{sec:example}.  Importantly, we can express
vertical composition of cases with a special method (here, $\mathit{Open}$) as:
\begin{align*}
  \mathtt{object}~O
  &=
  \lambda \{ O \mid \mathit{self}~Open \to \lambda\{x \to O \ask x\} \}
  \\
  \mathit{compose}
  &=
  \lambda o ~ o' \to
  \mathtt{object}~\{\ask x \to o.Open(o'.Open ~ x)\}
\end{align*}
so that
$\mathit{compose}~\mathtt{object}\{O\}~\mathtt{object}\{O'\} =
\mathtt{object}\{O \mid O'\}$.

% Syntactic sugar: $\Raise = \ans k \to k$, so that
% $\den{\Raise} = \lambda k. k$.  This is useful for recognizing answers of the
% form $E[\Raise]\ans$ which corresponds to raising a data structure
% corresponding to the question formed by refocusing the evaluation context $E$.

\subsection{Third refactor: Eliminating redundancy}

% Translating responses:
% \begin{align*}
%   \den{k}
%   &=
%   \lambda s.~ s ~ k
%   \\
%   \den{\varepsilon}
%   &=
%   \lambda s.~ s ~ ()
%   \\
%   \den{M \ans R}
%   &=
%   \lambda s.~ \den{R} ~ \lambda k.~ \den{M} ~ k ~ s
% \end{align*}
% Translating terms as functions from questions to answers:

At this point, there is some redundancy in the way the CPS translation handles
options.  Deriving a CPS transformation (here named $Opt$) for the option
handler $O \ask M$ gives:
\begin{align*}
  Opt\den{O \ask M} &= \lambda k.~ Opt\den{O} ~ k ~ \den{M} ~ k
\end{align*}
The given continuation is used twice: once to be analyzed for copattern
matching, and a copy to revert back on a failure so that $M$ can start again
from the original $k$.  However, it would be cleaner to just pass the
continuation once like so (naming the new transformation $Opt'$ to
disambiguate):
\begin{align*}
  % \den{x} &= x
  % \\
  % \den{M ~ X} &= \lambda k.~ \den{M} ~ (X ~ k)
  % \\
  % \den{M ~ N} &= \lambda k.~ \den{M} ~ (\den{N}, k)
  % \\
  % \den{M.} &= \lambda k.~ \den{M} ~ (\den{M}, k)
  % \\
  Opt'\den{O \ask M} &= \lambda k.~ Opt'\den{O} ~ \den{M} ~ k
  % \\
  % \den{\ans k \to R} &= \lambda k.~ \den{R}
\end{align*}
We can get away without the copy by modifying the failure term on each step of
copattern matching, in a way that restores the original structure of the
continuation as follows:
\begin{align*}
  Opt'\den{x \to O}
  &=
  \lambda f. \lambda k.
  \Case k \Of
  \\[-0.5ex]
  &\qquad
  \begin{alignedat}[t]{2}
    &(x, k') &&\to Opt'\den{O} ~ (\lambda q.~ f ~ (x, q)) ~ k'
    \\[-0.5ex]
    &k &&\to f ~ k
  \end{alignedat}
  \\
  Opt'\den{X \to O}
  &=
  \lambda f. \lambda k.
  \Case k \Of
  \\[-0.5ex]
  &\qquad
  \begin{alignedat}[t]{2}
    &(X ~ k') &&\to Opt'\den{O} ~ (\lambda q.~ f ~ (X ~ q)) ~ k'
    \\[-0.5ex]
    &k &&\to f ~ k
  \end{alignedat}
  \\
  Opt'\den{\ask x \to M}
  &=
  \lambda x. \den{M}
\end{align*}
Even though the failure handler $f$ is called with a different continuation, the
result is the same as before.
\begin{lemma}
  \begin{align*}
    Opt\den{O} ~ (q \compq k) ~ f ~ k
    =
    Opt'\den{O} ~ (\lambda k'.~ f ~ (q \compq k')) ~ k
  \end{align*}
\end{lemma}
\begin{proof}
  By induction on $L$ and cases on $k$.
\end{proof}
\begin{corollary}
  \begin{math}
    Opt\den{O} ~ k ~ f ~ k
    =
    Opt'\den{O} ~ f ~ k
    .
  \end{math}
\end{corollary}

\subsection{Fourth refactor: Fully continuation-passing style}

\begin{figure}
\centering
\footnotesize
\begin{alignat*}{2}
  \mathit{Response} \ni
  R
  &::= q
  \mid \varepsilon
  \mid M \ans R
  \\
  \mathit{Term} \ni M, N
  &::= x
  \mid M ~ N
  \mid M ~ X
  \mid M.
  \mid \Raise
  \mid O \ask M
  \mid \ans q \to R
  \\
  \mathit{Option} \ni O
  &::= x \to O
  \mid X \to O
  \mid \ask x \to M
\end{alignat*}
Translation of results $\den{R}$:
\begin{align*}
  \den{M \ans R} &= \lambda s.~ \den{R} ~ \lambda q.~ \den{M} ~ q ~ s
  &
  \den{q} &= \lambda s.~ s ~ q
  &
  \den{\varepsilon} &= \lambda s.~ s ~ ()
\end{align*}
Translation of terms $\den{M}$:
\begin{align*}
  \den{x} &= x
  \\
  \den{M~X} &= \lambda k.~ \den{M}~(X~k)
  \\
  \den{M~N} &= \lambda k.~ \den{M}~(\den{N}, k)
  \\
  \den{M.} &= \lambda k.~ \den{M}~(\den{M}, k)
  \\
  \den{\Raise} &= \lambda k.~ \lambda s. s ~ k
  \\
  \den{O \ask M} &= \lambda k.~ \den{O} ~ \den{M} ~ k
  \\
  \den{\ans q \to R} &= \lambda q.~ \den{R}
\end{align*}
Translation of options $\den{O}$:
\begin{align*}
  \den{x \to O}
  &=
  \lambda f. \lambda k.
  \begin{alignedat}[t]{2}
    \Case k \Of
    &(x, k') &&\to \den{O} ~ (\lambda q.~ f~(x, q)) ~ k'
    \\[-0.5ex]
    &k &&\to f ~ k
  \end{alignedat}
  \\
  \den{X \to O}
  &=
  \lambda f. \lambda k.
  \begin{alignedat}[t]{2}
    \Case k \Of
    &(X~k') &&\to \den{O} ~ (\lambda q.~ f~(X~q)) ~ k'
    \\[-0.5ex]
    &k &&\to f ~ k
  \end{alignedat}
  \\
  \den{\ask x \to M} &= \lambda x.~ \den{M}
\end{align*}
\caption{Refactored CPS and calculus of nested copatterns.}
\label{fig:nested-copat-cps}
\end{figure}

The last small detail revolves around the fact that the CPS transformation is no
longer strictly in CPS form, due to responses like $M \ans (N \ans R)$, which
gets transformed to an application of $\den{M}$ to the non-value
$\den{N}~\den{R}$.  But thankfully there is an easy fix to get back into strict
CPS: iterate CPS again~\cite{AbstractingControl}!  This gives another layer of
continuation for responses to elaborate the evaluation order of $M \ans R$ to
say $R$ is evaluated first, and on a return its answer is passed to $M$.  The
only affected cases of the transformation are:
\begin{align*}
  \den{M \ans R}
  &=
  \lambda s. \den{R} ~ \lambda k. \den{M} ~ k ~ s
  &
  \den{q}
  &=
  \lambda s. s ~ q
  &
  \den{\varepsilon}
  &=
  \lambda s. s ~ ()
  \\[1ex]
  \den{\Raise}
  &=
  \lambda k. \lambda s. s ~ k
\end{align*}

This last refactoring gives the final syntax and CPS transformation of the
compositional copattern calculus, which is shown in its entirety in
\cref{fig:nested-copat-cps}.

\section{Controlling Copatterns: The Return Voyage Back to Direct Style}
\label{sec:control-copat}

Having made the journey deriving semantics of monolithic copatterns from
small-step operational semantics to abstract machine to continuation-passing
style, we now aim to derive the semantics of compositional copatterns in reverse.

% \subsection{Preparing for the journey}

\subsection{Continuation-passing style}

\begin{figure}
\centering
\begin{haskell}
data Response i a
  = Splat a
  | End
  | Term i a :!: Response i a

data Term i a
  = Var a
  | Term i a :*: Term i a
  | Term i a :@: i
  | Dot (Term i a)
  | Option i a :?: Term i a
  | a :!-> Response i a
  | Raise

data Option i a
  = a :*-> Option i a
  | i :@-> Option i a
  | a :?-> Term i a
\end{haskell}
\caption{Data type representing refactored syntax trees.}
\label{fig:nest-syntax-code}
\end{figure}

\begin{figure}
\centering
\renewcommand{\basicstylesize}{\footnotesize}
\begin{haskell}
data Answer i a
  = Final   (CPSQuestion i a)
  | Stuck   [CPSTerm i a] a (CPSQuestion i a)
  | CoStuck [CPSTerm i a] a

run :: (Eq a, Eq i) => Response i a -> Answer i a
run r = (response (fmap Name r))

eval :: (Eq i, Eq a) => Term i a -> Answer i a
eval m = (term (fmap Name m)) Nop

try :: (Eq i, Eq a) => Option i a -> Answer i a
try o = (option (fmap Name o)) (term Raise) Nop

response :: (Eq i, Eq a) => Response i (CPSVar i a)
         -> Answer i a
response (Splat (Name k)) = CoStuck [] k
response (Splat (CPSQ q)) = Final q
response (Splat (CPST _)) = error "Illegal expr"
response (End)     = Final Nop
response (m :!: r) = (term m) <!> (response r)

(<!>) :: CPSTerm i a -> Answer i a -> Answer i a
f <!> Final q     = f q
f <!> Stuck gs x q = Stuck (f:gs) x q
f <!> CoStuck gs q = CoStuck (f:gs) q

term :: (Eq i, Eq a) => Term i (CPSVar i a)
     -> CPSTerm i a
term (Var (Name x)) = Stuck [] x
term (Var (CPST m)) = m
term (Var (CPSQ _)) = error "Illegal expr"
term (Dot m)    = \k -> (term m) (Arg(term m):*k)
term (m :*: n)  = \k -> (term m) (Arg(term n):*k)
term (m :@: i)  = \k -> (term m) (i :@ k)
term (Raise)    = \k -> Final k
term (q :!-> r) = \k -> (response r')
  where r' = r /!/ [(q, subQ k)]
term (o :?: m)  = \k -> (option o) (term m) k

option :: (Eq a, Eq i) => Option i (CPSVar i a)
       -> CPSOption i a
option (x :?-> m) = \f -> (term m')
  where m' = m // [(x, subT f)]
option (x :*-> o) = \f -> \case
  (y :* k) -> (option o') (f . (y :*)) k
  k        -> f k
  where o' = o /?/ [(x, subT $ useArg y)]
option (i :@-> o) = \f -> \case
  (j :@ k) | i == j -> (option o) (f . (i :@)) k
  k                 -> f k

subT m = TSub (Var (CPST m))
subQ k = RSub (Splat (CPSQ k))
\end{haskell}
\caption{Continuation-passing style translation of copatterns with nested options and control into Haskell.}
\label{fig:nest-cps-code}
\end{figure}

We begin with the continuation-passing style transformation from
\cref{fig:nested-copat-cps} that defines the semantics.  The high-level
specification of the transformation can be given a concrete Haskell
implementation shown in \cref{fig:nest-syntax-code,fig:nest-cps-code}.

\subsection{Abstract machine}

\begin{figure}
\centering
\renewcommand{\basicstylesize}{\footnotesize}
\begin{haskell}
data Answer i a
  = Final   (Cont i a)
  | Stuck   (MetaCont i a) a (Cont i a)
  | CoStuck (MetaCont i a) a

type Cont     i a = Question i a
type MetaCont i a = [Term i a]

subQ :: Question i a -> TRSub i a
subQ k = RSub (Raise `ask` k :!: End)

run :: (Eq a, Eq i) => Response i a
    -> Answer i a
run r = delim r []

eval :: (Eq i, Eq a) => Term i a
     -> Answer i a
eval m = refocus m Nop []

try :: (Eq i, Eq a) => Option i a
    -> Answer i a
try o = comatch o Raise Nop []

delim :: (Eq a, Eq i)
      => Response i a -> MetaCont i a
      -> Answer i a
delim (Splat k) s     = CoStuck s k
delim (End)     []    = Final Nop
delim (End)     (m:s) = refocus m Nop s
delim (m :!: r) s     = delim r (m : s)

refocus :: (Eq i, Eq a)
        => Term i a -> Cont i a -> MetaCont i a
        -> Answer i a
refocus (Var x)    k s = Stuck s x k
refocus (Dot m)    k s = refocus m (m:*k) s
refocus (m :*: n)  k s = refocus m (n:*k) s
refocus (m :@: i)  k s = refocus m (i:@k)s 
refocus (q :!-> r) k s = delim r' s
  where r' = r /!/ [(q, subQ k)]
refocus (o :?: m)  k s = comatch o m k s
refocus (Raise)    k (m:s) = refocus m k s
refocus (Raise)    k []    = Final k

comatch :: (Eq a, Eq i)
        => Option i a -> Term i a
        -> Cont i a -> MetaCont i a
        -> Answer i a
comatch (x :?-> n) m k      = refocus n' k
  where n' = n // [(x, TSub m)]
comatch (x :*-> o) m (n:*k) = comatch o' (m:*:n) k
  where o' = o /?/ [(x, TSub n)]
comatch (i :@-> o) m (j:@k)
  | i == j                  = comatch o (m:@:i) k
comatch o          m k      = refocus m k
\end{haskell}
\caption{Abstract machine for controlling copatterns.}
\label{fig:nest-machine-code}
\end{figure}

The continuation-passing style transformation can be modified in several
standard steps:
\begin{enumerate*}[(1)]
\item defunctionalization, 
\item waiting to apply the transformation function until the last moment,
\item syntactically inlining the semantics for substituted variables,
\item fusing substitution inlining with transformation,
\item $\eta$-expansion, and
\item rewriting chains of case-analysis as nested patterns.
\end{enumerate*}
The result of these program transformations on the CPS implementation gives the
tail-recursive abstract machine shown in \cref{fig:nest-machine-code}.  This
implementation can be rephrased into a traditional presentation of a stepping
relation on machine states:

\begin{itemize}
\item Delimiting steps:
\begin{align*}
  \braket{M \ans R \cmid S}
  &\mapsto
  \braket{R \cmid M; S}
  &
  \braket{\varepsilon \cmid M; S}
  &\mapsto
  \braket{M \cmid \varepsilon \cmid S}
\end{align*}

\item Refocusing steps:
  \begin{align*}
  % \braket{x \cmid K \cmid S}
  % &\not\mapsto
  % \\
  \braket{M ~ X \cmid K \cmid S}
  &\mapsto
  \braket{M \cmid X ~ K \cmid S}
  \\
  \braket{M ~ N \cmid K \cmid S}
  &\mapsto
  \braket{M \cmid N ~ K \cmid S}
  \\
  \braket{O \ask M \cmid K \cmid S}
  &\mapsto
  \braket{O \cmid M \cmid K \cmid S}
\end{align*}

\item Reduction steps:
  \begin{align*}
  \braket{M. \cmid K \cmid S}
  &\mapsto
  \braket{M \cmid M ~ K \cmid S}
  \\
  \braket{{\Raise} \cmid K \cmid M; S}
  &\mapsto
  \braket{M \cmid K \cmid S}
  \\
  \braket{\ans q \to R \cmid K \cmid S}
  &\mapsto
  \braket{R\subst{q}{K[\Raise]\ans\varepsilon} \cmid S}
\end{align*}

\item Copattern matching steps:
  \begin{align*}
  \braket{x \to O \cmid M \cmid N~K \cmid S}
  &\mapsto
  \braket{O\subst{x}{N} \cmid M~N \cmid K \cmid S}
  % \\
  % \braket{x \to O \cmid M \cmid K \cmid S}
  % &\mapsto
  % \braket{M \cmid K \cmid S}
  % &(K \neq N~K')
  \\
  \braket{X \to O \cmid M \cmid X~K \cmid S}
  &\mapsto
  \braket{O \cmid M~X \cmid K \cmid S}
  % \\
  % \braket{X \to O \cmid M \cmid K \cmid S}
  % &\mapsto
  % \braket{M \cmid K \cmid S}
  % &(K \neq X~K')
  \\
  \braket{\ask x \to N \cmid M \cmid K \cmid S}
  &\mapsto
  \braket{N\subst{x}{M} \cmid K \cmid S}
  \\
  \braket{O \cmid M \cmid K \cmid S}
  &\mapsto
  \braket{M \cmid K \cmid S}
  \quad(\text{otherwise})
\end{align*}
\end{itemize}

\begin{remark}
  \label{rm:environment-machine2}

  Similar to \cref{rm:environment-machine}, we can derive a more efficient
  environment-based machine by starting with a CPS that threads a substitution
  environment to lookup variables.  This machine is discussed in appendix
  \cref{sec:environment-machine}.
\end{remark}

\subsection{Small-step operational semantics}

\begin{figure}
\centering
\renewcommand{\basicstylesize}{\footnotesize}
\begin{haskell}
data CoFrame i a = Arg a | At i

data CoObject i a
  = CoO { coframe :: CoFrame i a,
          success :: Option i a,
          failure :: Term i a }

data RxTerm i a
  = FreeVar a
  | Introspect (Term i a)
  | Try a (Term i a) (Term i a)
  | Pop (CoObject i a) (Term i a)
  | Get (CoObject i a) i

data RdTerm i a
  = RdT (Term i a)
  | UnknownA a

reduce :: (Eq i, Eq a) => RxTerm i a
       -> RdTerm i a
reduce (Introspect m) = RdT $ m :*: m
reduce (Try x n m) = RdT $ n // [(x, TSub m)]
reduce (Pop (CoO (Arg x) o m) n)
                   = RdT $ o' :?: (m :*: n)
  where o' = o /?/ [(x, TSub m)]
reduce (Pop o n)   = RdT $ failure o :*: n
reduce (Get (CoO (At i) o m) j)
  | i == j         = RdT $ o :?: (m :@: i)
reduce (Get o i)   = RdT $ failure o :@: i
reduce (FreeVar x) = UnknownA x

data RxResponse i a
  = FreeCoVar a
  | Reset (Term i a) (Question i a)
  | Shift a (Response i a) (Question i a)
  | Under (CoObject i a)

data RdResponse i a
  = RdR (Response i a)
  | UnknownQ a

handle :: Eq a => RxResponse i a
       -> RdResponse i a
handle (FreeCoVar k) = UnknownQ k
handle (Reset m q)   = RdR $ m`ask`q:!:End
handle (Shift k r q) = RdR $ r/!/[(k, subQ q)]
handle (Under o)     = RdR $ failure o:!:End

subQ :: Question i a -> TRSub i a
subQ k = RSub (Raise `ask` k :!: End)
\end{haskell}
\caption{Functional small-step reduction of copatterns with delimited control.}
\label{fig:nest-reduce-code}
\end{figure}

\begin{figure}
\centering
\renewcommand{\basicstylesize}{\footnotesize}
\begin{haskell}
data Delimit i a
  = Around (RxTerm i a) (Question i a) [Term i a]
  | Caught (RxResponse i a) [Term i a]
  | Uncaught (Question i a)

delimit :: Response i a -> Delimit i a
delimit r = delim r []

unwind :: [Term i a] -> Response i a
       -> Response i a
unwind []    r = r
unwind (m:s) r = m :!: r

delim :: Response i a -> [Term i a]
      -> Delimit i a
delim (m :!: r) s = delim r (m : s)
delim (End) (m:s) = catch (refocus m Nop) s
delim (End) []    = Uncaught Nop
delim (Splat k) s = Caught (FreeCoVar k) s

catch :: Decomp i a -> [Term i a] -> Delimit i a
catch (Internal r q) s = Around r q s
catch (External r)   s = Caught r s
catch (Raised q) (m:s) = Caught (Reset m q) s
catch (Raised q) []    = Uncaught q

data Decomp i a
  = Internal (RxTerm i a) (Question i a)
  | External (RxResponse i a)
  | Raised   (Question i a)

decomp :: Term i a -> Decomp i a
decomp m = refocus m Nop

recomp :: Question i a -> Term i a -> Term i a
recomp q m = m `ask` q

refocus :: Term i a -> Question i a -> Decomp i a
refocus (m :*: n)  k = refocus m (n :* k)
refocus (m :@: i)  k = refocus m (i :@ k) 
refocus (o :?: m)  k = decide (consider o m) k
refocus (Dot m)    k = Internal (Introspect m) k
refocus (q :!-> r) k = External (Shift q r k)
refocus (Raise)    k = Raised k
refocus (Var x)    k = Internal (FreeVar x) k

data Consider i a
  = Inward a (Term i a) (Term i a)
  | Outward (CoObject i a)

only :: Option i a -> Consider i a
only o = consider o Raise

consider :: Option i a -> Term i a
         -> Consider i a
consider (x :?-> n) m = Inward x n m
consider (x :*-> o) m = Outward $ CoO (Arg x) o m
consider (i :@-> o) m = Outward $ CoO (At i) o m

decide :: Consider i a -> Question i a
       -> Decomp i a
decide (Outward o)    = comatch o
decide (Inward x n m) = Internal (Try x n m)

comatch :: CoObject i a -> Question i a
        -> Decomp i a
comatch o (n :* k) = Internal (Pop o n) k
comatch o (j :@ k) = Internal (Get o j) k
comatch o Nop      = External $ Under o
\end{haskell}
\caption{Decomposition of terms with delimited questions.}
\label{fig:nest-decomp-code}
\end{figure}

\begin{figure}
\centering
\renewcommand{\basicstylesize}{\footnotesize}
\begin{haskell}
data Answer i a
  = Final   (Question i a)
  | Stuck   [Term i a] a (Question i a)
  | CoStuck [Term i a] a

try :: (Eq i, Eq a) => Option i a -> Answer i a
try o = eval $ o :?: Raise

eval :: (Eq i, Eq a) => Term i a -> Answer i a
eval m = run $ m :!: End

run :: (Eq a, Eq i) => Response i a -> Answer i a
run r = case delimit r of
  Around r q s -> case reduce r of
    UnknownA x -> Stuck s x q
    RdT m      -> run $ unwind s $
                  recomp q m :!: End
  Caught r s   -> case handle r of
    UnknownQ k -> CoStuck s k
    RdR r      -> run $ unwind s r
  Uncaught q   -> Final q
\end{haskell}
\caption{Direct-style, delimited small-step interpreter.}
\label{fig:nest-small-step-code}
\end{figure}

The hardest step of the journey is to derive a small-step operational semantics
from the abstract machine.  This requires undoing several steps (fusion,
deforesting) which destroy information.  However, having already completed the
easier direction for a similar calculus, we have the advantage of knowing
something about the overall structures we should be looking for.  First,
identifying which steps are associated with reduction (the ones that can delete
or duplicate information), we can factor out a non-recursive \hs|reduce|
function that turns redexes into reducts on \hs|Term|s, along with a similar
\hs|handle| function on \hs|Response|s, shown in \cref{fig:nest-reduce-code}.
Next, by identifying the remaining steps that are purely refocusing (the ones
that are perfectly reversible), we can factor out a set of decomposition
functions that work through delimiters as shown in \cref{fig:nest-decomp-code}.
Finally, we can de-optimize refocusing in terms of decompose-recompose, and
catching delimited responses in terms of delimit-unwind.  This gives us the
direct-style, small-step operational interpreter shown in
\cref{fig:nest-small-step-code}.  The Haskell implementation can be
reinterpreted as the following corresponding small-step reduction relation:
\begin{alignat*}{3}
  \mathit{CoObj} &\ni{}& P
  &::= x \to O
  \mid X \to O
  \\
  % \mathit{Delimiter} &\ni{}& d
  % &::= q
  % \mid \varepsilon
  % \\
  \mathit{DelimCxt} &\ni{}& D
  &::= \hole
  \mid M \ans D
  % \mid E \ans \varepsilon
  \\
  \mathit{EvalCxt} &\ni{}& E
  &::= \hole
  \mid E ~ N
  \mid E ~ X
\end{alignat*}
\begin{gather*}
  \infer%[\mathit{Eval}]
  {E[M] \mapsto E[M']}
  {M \mapsto M'}
  \qquad
  \infer%[\mathit{Delim}]
  {M \ans \varepsilon \mapsto M' \ans \varepsilon}
  {M \mapsto M'}
  \qquad
  \infer%[\mathit{Handle}]
  {D[R] \mapsto D[R']}
  {R \mapsto R'}
\end{gather*}
\begin{align*}
  (\ask x \to N) \ask M
  &\mapsto
  N\subst{x}{M}
  \\
  ((x \to O) \ask M) ~ N
  &\mapsto
  O\subst{x}{N} \ask (M ~ N)
  \\
  ((X \to O) \ask M) ~ X
  &\mapsto
  O \ask (M ~ X)
  \\
  (P \ask M) ~ X
  &\mapsto
  M ~ X
  &(\text{otherwise})
  \\
  (P \ask M) ~ N
  &\mapsto
  M ~ N
  &(\text{otherwise})
  \\\\
  E[\ans k \to R] \ans \varepsilon
  &\mapsto
  R\subst{k}{(E[\Raise] \ans \varepsilon)}
  \\
  M \ans (E[\Raise] \ans \varepsilon)
  &\mapsto
  E[M] \ans \varepsilon
  \\
  (P \ask M) \ans \varepsilon
  &\mapsto
  M \ans \varepsilon
\end{align*}

Again, we get an analogous correspondence from the semantic derivations as
before in \cref{thm:block-semantic-correspondence}.  Delimited questions and
$\Raise$ make it more clear that the observable results of responses are
unanswered questions.  For simplicity, we focus on non-empty questions which
must be explicitly $\Raise$d.

\begin{theorem}
\label{thm:nest-semantic-correspondence}

The three \hs|eval| functions are equal, \ie the following relations between $R$
and $K \neq \varepsilon$ are all equivalent:
\begin{enumerate}[(a)]
\item $R \mapsto^* K[\Raise]\ans$.
\item $\cut{R}{\varepsilon} \mapsto^* \braket{{\Raise} \cmid K \cmid \varepsilon}$.
\item $\den{R}(\lambda k. k) \mapsto^* \den{K}$.
\end{enumerate}
\end{theorem}

\section{Related Work}
\label{sec:related-work}

\subsubsection*{Copatterns}
Our starting point comes from a macro implementation in Scheme and Racket
\cite{CoScheme}, where we are primarily concerned with specifying the behavior
of different dimensions of compositionality.  Alternative macro implementations
of copatterns have been given for OCaml \cite{LaforgueR17,jeannin_cocaml_2017},
which leverage different restrictions to aid code generation.  Copatterns have
also seen use in proof assistants like Agda
\cite{ElaboratingDependentCopatterns} which use a type-driven approach to
elaboration \cite{UnnestingCopatterns,ThibodeauMasters}.

\subsubsection*{Functional and syntactic correspondence}
The functional and syntactic correspondence between semantic artifacts
\cite{FunctionalCorrespondence,SyntacticCorrespondence,SmallStepBigStepMachines,DefunctionalizedInterpreters,InterDerivingSemanticArtifactsOOP,WalkInTheSemanticPark}
is based on the approach of definitional interpreters
\cite{DefinitionalInterpreters} and a long history of semantics-preserving
program transformations.  It has been especially useful for studying the
semantics of call-by-need evaluation
\cite{DefunctionalizedInterpretersCBNeed,InterDerivingSmallStepBigStepCBNeed,SyntheicOperationalCBNeed}
and its connection to sequent \cite{ADHNS2012CCSC} and process calculi
\cite{DMAV2014CPS}.

\subsubsection*{Delimited control}
The CPS semantics of delimited copattern matching is similar to delimited
control, specifically shift and reset
\cite{FunctionalAbstractionTypedContexts,AbstractingControl}.  We use a
call-by-name semantics for a close connection between evaluation contexts and
copatterns.  A similar approach to call-by-name delimited control
\cite{HerbelinG08} is related to shift0 \cite{DownenAriola2014CSCC}, a powerful
variant of shift \cite{materzok2011subtyping,materzok2012dynamic}.

\section{Conclusion}
\label{sec:conclusion}

Now at the end of our round-trip journey, the disciplined approach to deriving
semantic artifacts has been a powerful methodology for understanding complex
programming features.  The ability to generate CPS transformations is especially
useful to explore and refactor the language design space, and coming back gives
tools to understand source programs directly.  These artifacts still have
untapped potential to explore for understanding copattern and program
composition, such as extracting a type system from the CPS~\cite{FunctionalAbstractionTypedContexts}.

%%
%% The acknowledgments section is defined using the "acks" environment
%% (and NOT an unnumbered section). This ensures the proper
%% identification of the section in the article metadata, and the
%% consistent spelling of the heading.
\begin{acks}
  I want to thank Olivier Danvy for so generously devoting his time,
  encouragement, and teaching while I was still an early Ph.D. student.  While
  on an extended visit to the University of Oregon, Olivier gave a week-long
  hands-on tutorial on his technique for inter-deriving semantics.  Soon
  thereafter, I realized this was the perfect solution to a difficult problem on
  the semantics of the sequent calculus that had been lingering for over a year,
  which directly led to my second publication~\cite{ADHNS2012CCSC}.  Now, many
  years later, I had been puzzling over a similar problem of trying to capture
  the direct-style operational semantics for composing copatterns, which eluded
  me for quite some time.  Only after thinking of how I could contribute to
  OlivierFest, did I realize ``Aha!  This is the perfect problem to fix with the
  technique Olivier taught me!''  This paper is a celebration and revival of
  that early influence.

  This material is based upon work supported by the National Science Foundation
  under Grant No. 2245516.
\end{acks}

% \clearpage

%%
%% The next two lines define the bibliography style to be used, and
%% the bibliography file.
\bibliographystyle{ACM-Reference-Format}
\bibliography{refs}

%%% -*-BibTeX-*-
%%% Do NOT edit. File created by BibTeX with style
%%% ACM-Reference-Format-Journals [18-Jan-2012].

\begin{thebibliography}{31}

%%% ====================================================================
%%% NOTE TO THE USER: you can override these defaults by providing
%%% customized versions of any of these macros before the \bibliography
%%% command.  Each of them MUST provide its own final punctuation,
%%% except for \shownote{} and \showURL{}.  The latter two
%%% do not use final punctuation, in order to avoid confusing it with
%%% the Web address.
%%%
%%% To suppress output of a particular field, define its macro to expand
%%% to an empty string, or better, \unskip, like this:
%%%
%%% \newcommand{\showURL}[1]{\unskip}   % LaTeX syntax
%%%
%%% \def \showURL #1{\unskip}           % plain TeX syntax
%%%
%%% ====================================================================

\ifx \showCODEN    \undefined \def \showCODEN     #1{\unskip}     \fi
\ifx \showISBNx    \undefined \def \showISBNx     #1{\unskip}     \fi
\ifx \showISBNxiii \undefined \def \showISBNxiii  #1{\unskip}     \fi
\ifx \showISSN     \undefined \def \showISSN      #1{\unskip}     \fi
\ifx \showLCCN     \undefined \def \showLCCN      #1{\unskip}     \fi
\ifx \shownote     \undefined \def \shownote      #1{#1}          \fi
\ifx \showarticletitle \undefined \def \showarticletitle #1{#1}   \fi
\ifx \showURL      \undefined \def \showURL       {\relax}        \fi
% The following commands are used for tagged output and should be
% invisible to TeX
\providecommand\bibfield[2]{#2}
\providecommand\bibinfo[2]{#2}
\providecommand\natexlab[1]{#1}
\providecommand\showeprint[2][]{arXiv:#2}

\bibitem[Abel et~al\mbox{.}(2013)]%
        {Copatterns}
\bibfield{author}{\bibinfo{person}{Andreas Abel}, \bibinfo{person}{Brigitte
  Pientka}, \bibinfo{person}{David Thibodeau}, {and} \bibinfo{person}{Anton
  Setzer}.} \bibinfo{year}{2013}\natexlab{}.
\newblock \showarticletitle{Copatterns: Programming Infinite Structures by
  Observations}. In \bibinfo{booktitle}{\emph{Proceedings of the 40th Annual
  {ACM} {SIGPLAN-SIGACT} Symposium on Principles of Programming Languages}}
  (Rome, Italy) \emph{(\bibinfo{series}{POPL~'13})}. \bibinfo{publisher}{ACM},
  \bibinfo{address}{New York, NY, USA}, \bibinfo{pages}{27--38}.
\newblock
\showISBNx{978-1-4503-1832-7}
\href{https://doi.org/10.1145/2429069.2429075}{doi:\nolinkurl{10.1145/2429069.2429075}}


\bibitem[Ager et~al\mbox{.}(2003)]%
        {FunctionalCorrespondence}
\bibfield{author}{\bibinfo{person}{Mads~Sig Ager}, \bibinfo{person}{Dariusz
  Biernacki}, \bibinfo{person}{Olivier Danvy}, {and} \bibinfo{person}{Jan
  Midtgaard}.} \bibinfo{year}{2003}\natexlab{}.
\newblock \showarticletitle{A functional correspondence between evaluators and
  abstract machines}. In \bibinfo{booktitle}{\emph{Proceedings of the 5th ACM
  SIGPLAN International Conference on Principles and Practice of Declaritive
  Programming}} (Uppsala, Sweden) \emph{(\bibinfo{series}{PPDP '03})}.
  \bibinfo{publisher}{ACM}, \bibinfo{address}{New York, NY, USA},
  \bibinfo{pages}{8–19}.
\newblock
\showISBNx{1581137052}
\href{https://doi.org/10.1145/888251.888254}{doi:\nolinkurl{10.1145/888251.888254}}


\bibitem[Ariola et~al\mbox{.}(2012)]%
        {ADHNS2012CCSC}
\bibfield{author}{\bibinfo{person}{Zena~M. Ariola}, \bibinfo{person}{Paul
  Downen}, \bibinfo{person}{Hugo Herbelin}, \bibinfo{person}{Keiko Nakata},
  {and} \bibinfo{person}{Alexis Saurin}.} \bibinfo{year}{2012}\natexlab{}.
\newblock \showarticletitle{Classical Call-by-Need Sequent Calculi: The Unity
  of Semantic Artifacts}.
\newblock In \bibinfo{booktitle}{\emph{Functional and Logic Programming: 11th
  International Symposium}}. Vol.~\bibinfo{volume}{7294}.
  \bibinfo{publisher}{Springer Berlin Heidelberg}, \bibinfo{address}{Berlin,
  Heidelberg}, \bibinfo{pages}{32--46}.
\newblock
\showISBNx{978-3-642-29821-9}
\href{https://doi.org/10.1007/978-3-642-29822-6_6}{doi:\nolinkurl{10.1007/978-3-642-29822-6_6}}


\bibitem[Biernacka and Danvy(2007)]%
        {SyntacticCorrespondence}
\bibfield{author}{\bibinfo{person}{Magorzata Biernacka} {and}
  \bibinfo{person}{Olivier Danvy}.} \bibinfo{year}{2007}\natexlab{}.
\newblock \showarticletitle{A syntactic correspondence between
  context-sensitive calculi and abstract machines}.
\newblock \bibinfo{journal}{\emph{Theoretical Computer Science}}
  \bibinfo{volume}{375}, \bibinfo{number}{1–3} (\bibinfo{date}{April}
  \bibinfo{year}{2007}), \bibinfo{pages}{76–108}.
\newblock
\showISSN{0304-3975}
\href{https://doi.org/10.1016/j.tcs.2006.12.028}{doi:\nolinkurl{10.1016/j.tcs.2006.12.028}}


\bibitem[Cockx and Abel(2018)]%
        {ElaboratingDependentCopatterns}
\bibfield{author}{\bibinfo{person}{Jesper Cockx} {and} \bibinfo{person}{Andreas
  Abel}.} \bibinfo{year}{2018}\natexlab{}.
\newblock \showarticletitle{Elaborating dependent (co)pattern matching}.
\newblock \bibinfo{journal}{\emph{Proceedings of the ACM on Programming
  Languages}} \bibinfo{volume}{2}, \bibinfo{number}{ICFP}, Article
  \bibinfo{articleno}{75} (\bibinfo{year}{2018}), \bibinfo{numpages}{30}~pages.
\newblock
\href{https://doi.org/10.1145/3236770}{doi:\nolinkurl{10.1145/3236770}}


\bibitem[Danvy(2008)]%
        {DefunctionalizedInterpreters}
\bibfield{author}{\bibinfo{person}{Olivier Danvy}.}
  \bibinfo{year}{2008}\natexlab{}.
\newblock \showarticletitle{Defunctionalized interpreters for programming
  languages}. In \bibinfo{booktitle}{\emph{Proceedings of the 13th ACM SIGPLAN
  International Conference on Functional Programming}} (Victoria, BC, Canada)
  \emph{(\bibinfo{series}{ICFP '08})}. \bibinfo{publisher}{ACM},
  \bibinfo{address}{New York, NY, USA}, \bibinfo{pages}{131–142}.
\newblock
\showISBNx{9781595939197}
\href{https://doi.org/10.1145/1411204.1411206}{doi:\nolinkurl{10.1145/1411204.1411206}}


\bibitem[Danvy and Filinski(1989)]%
        {FunctionalAbstractionTypedContexts}
\bibfield{author}{\bibinfo{person}{Olivier Danvy} {and}
  \bibinfo{person}{Andrzej Filinski}.} \bibinfo{year}{1989}\natexlab{}.
\newblock \bibinfo{booktitle}{\emph{A Functional Abstraction of Typed
  Contexts}}.
\newblock \bibinfo{type}{{T}echnical {R}eport} 89/12.
  \bibinfo{institution}{DIKU, University of Copenhagen, Copenhagen, Denmark}.
\newblock


\bibitem[Danvy and Filinski(1990)]%
        {AbstractingControl}
\bibfield{author}{\bibinfo{person}{Olivier Danvy} {and}
  \bibinfo{person}{Andrzej Filinski}.} \bibinfo{year}{1990}\natexlab{}.
\newblock \showarticletitle{Abstracting Control}. In
  \bibinfo{booktitle}{\emph{Proceedings of the 1990 {ACM} Conference on {LISP}
  and Functional Programming, {LFP} 1990, Nice, France, 27-29 June 1990}}.
  \bibinfo{publisher}{{ACM}}, \bibinfo{address}{New York, NY, USA},
  \bibinfo{pages}{151--160}.
\newblock
\href{https://doi.org/10.1145/91556.91622}{doi:\nolinkurl{10.1145/91556.91622}}


\bibitem[Danvy and Johannsen(2010)]%
        {InterDerivingSemanticArtifactsOOP}
\bibfield{author}{\bibinfo{person}{Olivier Danvy} {and} \bibinfo{person}{Jacob
  Johannsen}.} \bibinfo{year}{2010}\natexlab{}.
\newblock \showarticletitle{Inter-deriving semantic artifacts for
  object-oriented programming}.
\newblock \bibinfo{journal}{\emph{J. Comput. System Sci.}}
  \bibinfo{volume}{76}, \bibinfo{number}{5} (\bibinfo{date}{Aug.}
  \bibinfo{year}{2010}), \bibinfo{pages}{302–323}.
\newblock
\showISSN{0022-0000}
\href{https://doi.org/10.1016/j.jcss.2009.10.004}{doi:\nolinkurl{10.1016/j.jcss.2009.10.004}}


\bibitem[Danvy et~al\mbox{.}(2011)]%
        {WalkInTheSemanticPark}
\bibfield{author}{\bibinfo{person}{Olivier Danvy}, \bibinfo{person}{Jacob
  Johannsen}, {and} \bibinfo{person}{Ian Zerny}.}
  \bibinfo{year}{2011}\natexlab{}.
\newblock \showarticletitle{A walk in the semantic park}. In
  \bibinfo{booktitle}{\emph{Proceedings of the 20th ACM SIGPLAN Workshop on
  Partial Evaluation and Program Manipulation}} \emph{(\bibinfo{series}{PEPM
  '11})}. \bibinfo{publisher}{ACM}, \bibinfo{address}{New York, NY, USA},
  \bibinfo{pages}{1–12}.
\newblock
\showISBNx{9781450304856}
\href{https://doi.org/10.1145/1929501.1929503}{doi:\nolinkurl{10.1145/1929501.1929503}}


\bibitem[Danvy and Millikin(2008)]%
        {SmallStepBigStepMachines}
\bibfield{author}{\bibinfo{person}{Olivier Danvy} {and} \bibinfo{person}{Kevin
  Millikin}.} \bibinfo{year}{2008}\natexlab{}.
\newblock \showarticletitle{On the equivalence between small-step and big-step
  abstract machines: a simple application of lightweight fusion}.
\newblock \bibinfo{journal}{\emph{Inform. Process. Lett.}}
  \bibinfo{volume}{106}, \bibinfo{number}{3} (\bibinfo{date}{April}
  \bibinfo{year}{2008}), \bibinfo{pages}{100–109}.
\newblock
\showISSN{0020-0190}
\href{https://doi.org/10.1016/j.ipl.2007.10.010}{doi:\nolinkurl{10.1016/j.ipl.2007.10.010}}


\bibitem[Danvy et~al\mbox{.}(2010)]%
        {DefunctionalizedInterpretersCBNeed}
\bibfield{author}{\bibinfo{person}{Olivier Danvy}, \bibinfo{person}{Kevin
  Millikin}, \bibinfo{person}{Johan Munk}, {and} \bibinfo{person}{Ian Zerny}.}
  \bibinfo{year}{2010}\natexlab{}.
\newblock \showarticletitle{Defunctionalized interpreters for call-by-need
  evaluation}. In \bibinfo{booktitle}{\emph{Proceedings of the 10th
  International Conference on Functional and Logic Programming}} (Sendai,
  Japan) \emph{(\bibinfo{series}{FLOPS'10})}.
  \bibinfo{publisher}{Springer-Verlag}, \bibinfo{address}{Berlin, Heidelberg},
  \bibinfo{pages}{240–256}.
\newblock
\showISBNx{3642122507}
\href{https://doi.org/10.1007/978-3-642-12251-4_18}{doi:\nolinkurl{10.1007/978-3-642-12251-4_18}}


\bibitem[Danvy et~al\mbox{.}(2012)]%
        {InterDerivingSmallStepBigStepCBNeed}
\bibfield{author}{\bibinfo{person}{Olivier Danvy}, \bibinfo{person}{Kevin
  Millikin}, \bibinfo{person}{Johan Munk}, {and} \bibinfo{person}{Ian Zerny}.}
  \bibinfo{year}{2012}\natexlab{}.
\newblock \showarticletitle{On inter-deriving small-step and big-step
  semantics: A case study for storeless call-by-need evaluation}.
\newblock \bibinfo{journal}{\emph{Theoretical Computer Science}}
  \bibinfo{volume}{435} (\bibinfo{date}{June} \bibinfo{year}{2012}),
  \bibinfo{pages}{21–42}.
\newblock
\showISSN{0304-3975}
\href{https://doi.org/10.1016/j.tcs.2012.02.023}{doi:\nolinkurl{10.1016/j.tcs.2012.02.023}}


\bibitem[Danvy and Zerny(2013)]%
        {SyntheicOperationalCBNeed}
\bibfield{author}{\bibinfo{person}{Olivier Danvy} {and} \bibinfo{person}{Ian
  Zerny}.} \bibinfo{year}{2013}\natexlab{}.
\newblock \showarticletitle{A synthetic operational account of call-by-need
  evaluation}. In \bibinfo{booktitle}{\emph{Proceedings of the 15th Symposium
  on Principles and Practice of Declarative Programming}}
  \emph{(\bibinfo{series}{PPDP '13})}. \bibinfo{publisher}{ACM},
  \bibinfo{address}{New York, NY, USA}, \bibinfo{pages}{97–108}.
\newblock
\showISBNx{9781450321549}
\href{https://doi.org/10.1145/2505879.2505898}{doi:\nolinkurl{10.1145/2505879.2505898}}


\bibitem[Downen(2025)]%
        {ControlCopatCode}
\bibfield{author}{\bibinfo{person}{Paul Downen}.}
  \bibinfo{year}{2025}\natexlab{}.
\newblock \bibinfo{title}{Semantic Artifacts for "Controlling Copatterns: There
  and Back Again"}.
\newblock
\href{https://doi.org/10.5281/zenodo.16888452}{doi:\nolinkurl{10.5281/zenodo.16888452}}


\bibitem[Downen and Ariola(2014)]%
        {DownenAriola2014CSCC}
\bibfield{author}{\bibinfo{person}{Paul Downen} {and} \bibinfo{person}{Zena~M.
  Ariola}.} \bibinfo{year}{2014}\natexlab{}.
\newblock \showarticletitle{Compositional Semantics for Composable
  Continuations: From Abortive to Delimited Control}. In
  \bibinfo{booktitle}{\emph{Proceedings of the 19th {ACM} {SIGPLAN}
  International Conference on Functional Programming}} (Gothenburg, Sweden)
  \emph{(\bibinfo{series}{ICFP~'14})}. \bibinfo{publisher}{ACM},
  \bibinfo{address}{New York, NY, USA}, \bibinfo{pages}{109--122}.
\newblock
\showISBNx{978-1-4503-2873-9}
\href{https://doi.org/10.1145/2628136.2628147}{doi:\nolinkurl{10.1145/2628136.2628147}}


\bibitem[Downen and Corbelino~II(2025)]%
        {CoScheme}
\bibfield{author}{\bibinfo{person}{Paul Downen} {and} \bibinfo{person}{Adriano
  Corbelino~II}.} \bibinfo{year}{2025}\natexlab{}.
\newblock \showarticletitle{CoScheme: Compositional Copatterns in Scheme}. In
  \bibinfo{booktitle}{\emph{International Symposium on Trends in Functional
  Programming}}. \bibinfo{publisher}{Springer}, \bibinfo{numpages}{37}~pages.
\newblock


\bibitem[Downen et~al\mbox{.}(2014)]%
        {DMAV2014CPS}
\bibfield{author}{\bibinfo{person}{Paul Downen}, \bibinfo{person}{Luke Maurer},
  \bibinfo{person}{Zena~M. Ariola}, {and} \bibinfo{person}{Daniele Varacca}.}
  \bibinfo{year}{2014}\natexlab{}.
\newblock \showarticletitle{Continuations, Processes, and Sharing}. In
  \bibinfo{booktitle}{\emph{Proceedings of the 16th International Symposium on
  Principles and Practice of Declarative Programming}} (Canterbury, United
  Kingdom) \emph{(\bibinfo{series}{PPDP~'14})}. \bibinfo{publisher}{ACM},
  \bibinfo{address}{New York, NY, USA}, \bibinfo{pages}{69--80}.
\newblock
\showISBNx{978-1-4503-2947-7}
\href{https://doi.org/10.1145/2643135.2643155}{doi:\nolinkurl{10.1145/2643135.2643155}}


\bibitem[Fujita(2003)]%
        {CBNSyntacticCPS}
\bibfield{author}{\bibinfo{person}{Ken-Etsu Fujita}.}
  \bibinfo{year}{2003}\natexlab{}.
\newblock \showarticletitle{A sound and complete {CPS}-translation for
  {$\lambda\mu$}-calculus}. In \bibinfo{booktitle}{\emph{Proceedings of the 6th
  International Conference on Typed Lambda Calculi and Applications}}
  (Valencia, Spain) \emph{(\bibinfo{series}{TLCA'03})}.
  \bibinfo{publisher}{Springer-Verlag}, \bibinfo{address}{Berlin, Heidelberg},
  \bibinfo{pages}{120–134}.
\newblock
\showISBNx{3540403329}
\href{https://doi.org/10.1007/3-540-44904-3_9}{doi:\nolinkurl{10.1007/3-540-44904-3_9}}


\bibitem[Herbelin and Ghilezan(2008)]%
        {HerbelinG08}
\bibfield{author}{\bibinfo{person}{Hugo Herbelin} {and} \bibinfo{person}{Silvia
  Ghilezan}.} \bibinfo{year}{2008}\natexlab{}.
\newblock \showarticletitle{An Approach to Call-by-Name Delimited
  Continuations}. In \bibinfo{booktitle}{\emph{Proceedings of the 35th Annual
  ACM SIGPLAN-SIGACT Symposium on Principles of Programming Languages}} (San
  Francisco, California, USA) \emph{(\bibinfo{series}{POPL~'08})}.
  \bibinfo{publisher}{ACM}, \bibinfo{address}{New York, NY, USA},
  \bibinfo{pages}{383–394}.
\newblock
\showISBNx{9781595936899}
\href{https://doi.org/10.1145/1328438.1328484}{doi:\nolinkurl{10.1145/1328438.1328484}}


\bibitem[Hofmann and Streicher(1997)]%
        {ContinuationModels}
\bibfield{author}{\bibinfo{person}{Martin Hofmann} {and}
  \bibinfo{person}{Thomas Streicher}.} \bibinfo{year}{1997}\natexlab{}.
\newblock \showarticletitle{Continuation models are universal for
  lambda-mu-calculus}. In \bibinfo{booktitle}{\emph{Proceedings of the 12th
  Annual IEEE Symposium on Logic in Computer Science}}
  \emph{(\bibinfo{series}{LICS '97})}. \bibinfo{publisher}{IEEE Computer
  Society}, \bibinfo{address}{USA}, \bibinfo{pages}{387}.
\newblock
\showISBNx{0818679255}
\href{https://doi.org/10.1109/LICS.1997.614964}{doi:\nolinkurl{10.1109/LICS.1997.614964}}


\bibitem[Jeannin et~al\mbox{.}(2017)]%
        {jeannin_cocaml_2017}
\bibfield{author}{\bibinfo{person}{Jean-Baptiste Jeannin},
  \bibinfo{person}{Dexter Kozen}, {and} \bibinfo{person}{Alexandra Silva}.}
  \bibinfo{year}{2017}\natexlab{}.
\newblock \showarticletitle{{CoCaml}: Functional Programming with Regular
  Coinductive Types}.
\newblock \bibinfo{journal}{\emph{Fundamenta Informaticae}}
  \bibinfo{volume}{150}, \bibinfo{number}{3} (\bibinfo{year}{2017}),
  \bibinfo{pages}{347--377}.
\newblock
\showISSN{01692968, 18758681}
\href{https://doi.org/10.3233/FI-2017-1473}{doi:\nolinkurl{10.3233/FI-2017-1473}}


\bibitem[Laforgue and R\'{e}gis-Gianas(2017)]%
        {LaforgueR17}
\bibfield{author}{\bibinfo{person}{Paul Laforgue} {and} \bibinfo{person}{Yann
  R\'{e}gis-Gianas}.} \bibinfo{year}{2017}\natexlab{}.
\newblock \showarticletitle{Copattern matching and first-class observations in
  {OCaml}, with a macro}. In \bibinfo{booktitle}{\emph{Proceedings of the 19th
  International Symposium on Principles and Practice of Declarative
  Programming}} (Namur, Belgium) \emph{(\bibinfo{series}{PPDP '17})}.
  \bibinfo{publisher}{ACM}, \bibinfo{address}{New York, NY, USA},
  \bibinfo{pages}{97–108}.
\newblock
\showISBNx{9781450352918}
\href{https://doi.org/10.1145/3131851.3131869}{doi:\nolinkurl{10.1145/3131851.3131869}}


\bibitem[Materzok and Biernacki(2011)]%
        {materzok2011subtyping}
\bibfield{author}{\bibinfo{person}{Marek Materzok} {and}
  \bibinfo{person}{Dariusz Biernacki}.} \bibinfo{year}{2011}\natexlab{}.
\newblock \showarticletitle{Subtyping Delimited Continuations}. In
  \bibinfo{booktitle}{\emph{Proceedings of the 16th ACM SIGPLAN International
  Conference on Functional Programming}} (Tokyo, Japan)
  \emph{(\bibinfo{series}{ICFP~'11})}. \bibinfo{publisher}{ACM},
  \bibinfo{address}{New York, NY, USA}, \bibinfo{pages}{81–93}.
\newblock
\showISBNx{9781450308656}
\href{https://doi.org/10.1145/2034773.2034786}{doi:\nolinkurl{10.1145/2034773.2034786}}


\bibitem[Materzok and Biernacki(2012)]%
        {materzok2012dynamic}
\bibfield{author}{\bibinfo{person}{Marek Materzok} {and}
  \bibinfo{person}{Dariusz Biernacki}.} \bibinfo{year}{2012}\natexlab{}.
\newblock \showarticletitle{A Dynamic Interpretation of the {CPS} Hierarchy}.
  In \bibinfo{booktitle}{\emph{Programming Languages and Systems}}.
  \bibinfo{publisher}{Springer Berlin Heidelberg}, \bibinfo{address}{Berlin,
  Heidelberg}, \bibinfo{pages}{296--311}.
\newblock
\showISBNx{978-3-642-35182-2}
\href{https://doi.org/10.1007/978-3-642-35182-2\_21}{doi:\nolinkurl{10.1007/978-3-642-35182-2\_21}}


\bibitem[Rendel et~al\mbox{.}(2015)]%
        {RefunctionalizationCopatterns}
\bibfield{author}{\bibinfo{person}{Tillmann Rendel}, \bibinfo{person}{Julia
  Trieflinger}, {and} \bibinfo{person}{Klaus Ostermann}.}
  \bibinfo{year}{2015}\natexlab{}.
\newblock \showarticletitle{Automatic refunctionalization to a language with
  copattern matching: with applications to the expression problem}. In
  \bibinfo{booktitle}{\emph{Proceedings of the 20th ACM SIGPLAN International
  Conference on Functional Programming}} (Vancouver, BC, Canada)
  \emph{(\bibinfo{series}{ICFP 2015})}. \bibinfo{publisher}{ACM},
  \bibinfo{address}{New York, NY, USA}, \bibinfo{pages}{269–279}.
\newblock
\showISBNx{9781450336697}
\href{https://doi.org/10.1145/2784731.2784763}{doi:\nolinkurl{10.1145/2784731.2784763}}


\bibitem[Reynolds(1972)]%
        {DefinitionalInterpreters}
\bibfield{author}{\bibinfo{person}{John~C. Reynolds}.}
  \bibinfo{year}{1972}\natexlab{}.
\newblock \showarticletitle{Definitional interpreters for higher-order
  programming languages}. In \bibinfo{booktitle}{\emph{Proceedings of the ACM
  Annual Conference - Volume 2}} (Boston, Massachusetts, USA)
  \emph{(\bibinfo{series}{ACM '72})}. \bibinfo{publisher}{ACM},
  \bibinfo{address}{New York, NY, USA}, \bibinfo{pages}{717–740}.
\newblock
\showISBNx{9781450374927}
\href{https://doi.org/10.1145/800194.805852}{doi:\nolinkurl{10.1145/800194.805852}}


\bibitem[Setzer et~al\mbox{.}(2014)]%
        {UnnestingCopatterns}
\bibfield{author}{\bibinfo{person}{Anton Setzer}, \bibinfo{person}{Andreas
  Abel}, \bibinfo{person}{Brigitte Pientka}, {and} \bibinfo{person}{David
  Thibodeau}.} \bibinfo{year}{2014}\natexlab{}.
\newblock \showarticletitle{Unnesting of Copatterns}. In
  \bibinfo{booktitle}{\emph{Rewriting and Typed Lambda Calculi - Joint
  International Conference, {RTA-TLCA} 2014, Held as Part of the Vienna Summer
  of Logic, Vienna, Austria, July 14-17, 2014. Proceedings}},
  Vol.~\bibinfo{volume}{8560}. \bibinfo{publisher}{Springer},
  \bibinfo{address}{Cham}, \bibinfo{pages}{31--45}.
\newblock
\href{https://doi.org/10.1007/978-3-319-08918-8\_3}{doi:\nolinkurl{10.1007/978-3-319-08918-8\_3}}


\bibitem[Streicher and Reus(1998)]%
        {ClassicalLogicContinuationSemantics}
\bibfield{author}{\bibinfo{person}{Th. Streicher} {and} \bibinfo{person}{B.
  Reus}.} \bibinfo{year}{1998}\natexlab{}.
\newblock \showarticletitle{Classical logic, continuation semantics and
  abstract machines}.
\newblock \bibinfo{journal}{\emph{Journal of Functional Programming}}
  \bibinfo{volume}{8}, \bibinfo{number}{6} (\bibinfo{date}{Nov.}
  \bibinfo{year}{1998}), \bibinfo{pages}{543–572}.
\newblock
\showISSN{0956-7968}
\href{https://doi.org/10.1017/S0956796898003141}{doi:\nolinkurl{10.1017/S0956796898003141}}


\bibitem[Thibodeau(2015)]%
        {ThibodeauMasters}
\bibfield{author}{\bibinfo{person}{David Thibodeau}.}
  \bibinfo{year}{2015}\natexlab{}.
\newblock \emph{\bibinfo{title}{Programming Infinite Structures using
  Copatterns}}.
\newblock \bibinfo{thesistype}{Master's\ thesis}. \bibinfo{school}{School of
  Computer Science, Mcgill University, Montreal}.
\newblock


\bibitem[Thielecke(2004)]%
        {AnswerTypePolyEta}
\bibfield{author}{\bibinfo{person}{Hayo Thielecke}.}
  \bibinfo{year}{2004}\natexlab{}.
\newblock \showarticletitle{Answer Type Polymorphism in Call-by-Name
  Continuation Passing}. In \bibinfo{booktitle}{\emph{Programming Languages and
  Systems}}. \bibinfo{publisher}{Springer Berlin Heidelberg},
  \bibinfo{address}{Berlin, Heidelberg}, \bibinfo{pages}{279--293}.
\newblock
\showISBNx{978-3-540-24725-8}
\href{https://doi.org/10.1007/978-3-540-24725-8_20}{doi:\nolinkurl{10.1007/978-3-540-24725-8_20}}


\end{thebibliography}

%%
%% If your work has an appendix, this is the place to put it.
% \clearpage
\appendix

\section{Postscript: Efficiently Passing Environments}
\label{sec:environment-machine}

%% Help LaTeX place figures effectively by allowing them all as early as possible
\begin{figure}
\centering
\renewcommand{\basicstylesize}{\footnotesize}
\begin{haskell}
type ClosEnv i a = Env a (Closure i a)
type ClosQuestion i a
  = Copattern i (Closure i a)
data Closure i a
  = (:/:) { openTerm  :: Term i a,
            staticEnv :: ClosEnv i a }

data Redex i a
  = Introspect (Term i a) (ClosEnv i a)
  | Respond [Option i a] (ClosEnv i a)
  | FreeVar a (ClosEnv i a)

data Reduct i a
  = Reduced (Closure i a)
  | Unhandled
  | Unknown a

reduce :: (Eq i, Eq a)
       => Redex i a -> ClosQuestion i a
       -> Followup i a
reduce (Introspect m env) q
  = Next (Reduced $ m :*: m :/: env) q
reduce (FreeVar x env)    q
  = case lookup x env of
  Nothing -> Next (Unknown x) q
  Just m  -> Next (Reduced m) q
reduce (Respond (lhs :-> rhs : ops) env) q
  = case suffix match of
      Followup q'  ->
        Next (Reduced $ rhs:/:env'++env) q'
      Unasked lhs' ->
        More lhs' (rhs:/:env') ops env q
      Mismatch _ _ ->
        reduce (Respond ops env) q
  where match = comatch lhs q
        env'  = prefix match
reduce (Respond [] env)   q
  = Next Unhandled q

data Decomp i a
  = Asked (Redex i a) (ClosQuestion i a)

decomp  :: Closure i a -> Decomp i a
recomp  :: Term i a -> Question i a -> Term i a
refocus :: Closure i a -> ClosQuestion i a
        -> Decomp i a

eval :: (Eq a, Eq i) => Term i a -> Answer i a
eval m = iter $ decomp (m :/: [])

iter :: (Eq a, Eq i) => Decomp i a -> Answer i a
iter (Asked r q) = case reduce r q of
  Next (Reduced m)     k -> iter $ refocus m k
  Next (Unknown x)     k -> Stuck x k
  Next Unhandled       k -> Raise k
  More lhs rhs ops env k -> Under lhs rhs ops env k
\end{haskell}
\caption{Small-step reduction with an environment}
\label{fig:block-env-reduce-code}
\end{figure}

\begin{figure}
\centering
\renewcommand{\basicstylesize}{\footnotesize}
\begin{haskell}
data Answer i a
  = Under (Copattern i a) (Closure i a)
          [Option i a] (ClosEnv i a)
          (ClosQuestion i a)
  | Raise (ClosQuestion i a)
  | Stuck a (ClosQuestion i a)

eval :: (Eq a, Eq i) => Term i a -> Answer i a
eval m = refocus m [] Nop

refocus :: (Eq a, Eq i) => Term i a
        -> ClosEnv i a -> ClosQuestion i a
        -> Answer i a
refocus (Var x)   env k = case lookup x env of
  Nothing -> Stuck x k
  Just (m :/: env)  -> refocus m env k
refocus (Dot m)   env k
  = refocus m env $ (m :/: env) :* k
refocus (Obj os)  env k = case os of
  lhs :-> rhs : os -> comatch lhs k [] rhs os env k
  []               -> Raise k
refocus (m :*: n) env k
  = refocus m env $ (n :/: env) :* k
refocus (m :@: i) env k
  = refocus m env $ i :@ k

comatch :: (Eq a, Eq i) => Copattern i a
        -> ClosQuestion i a -> ClosEnv i a
        -> Term i a -> [Option i a] -> ClosEnv i a
        -> ClosQuestion i a
        -> Answer i a
comatch Nop      cxt      env' rhs _  env _
  = refocus rhs (env' ++ env) cxt
comatch lhs      Nop      env' rhs os env q
  = Under lhs (rhs :/: env') os env q
comatch (x:*lhs) (y:*cxt) env' rhs os env q
  = comatch lhs cxt ((x,y):env') rhs os env q
comatch (i:@lhs) (j:@cxt) env' rhs os env q
  | i == j = comatch lhs cxt env' rhs os env q
comatch lhs      cxt      _    _   os env q
  = refocus (Obj os) env q
\end{haskell}
\caption{Environment-passing, tail-recursive abstract machine interpreter.}
\label{fig:block-env-machine-code}
\end{figure}

\begin{figure}
\centering
\renewcommand{\basicstylesize}{\footnotesize}
\begin{haskell}
data Answer i a
  = Final   (CPSQuestion i a)
  | Stuck   [CPSTerm i a] a (CPSQuestion i a)
  | CoStuck [CPSTerm i a] a

type CPSQuestion i a = Copattern i (CPSArg i a)
type CPSResponse i a = Answer i a
type CPSTerm i a = CPSQuestion i a -> CPSResponse i a
type CPSOption i a = CPSTerm i a -> CPSTerm i a

newtype CPSArg i a
  = Arg { useArg :: CPSTerm i a }

data CPSSub i a = CPST (CPSTerm i a)
                | CPSQ (CPSQuestion i a)
type CPSEnv i a = Env a (CPSSub i a)

run :: (Eq a, Eq i) => Response i a -> Answer i a
run r = (response r [])

eval :: (Eq i, Eq a) => Term i a
     -> Answer i a
eval m = (term m []) Nop

try :: (Eq i, Eq a) => Option i a -> Answer i a
try o = (option o []) Nop (term Raise []) Nop

response :: (Eq a, Eq i) => Response i a
         -> CPSEnv i a -> Answer i a
response (Splat k) env = case lookup k env of
  Just (CPSQ q) -> Final q
  _             -> CoStuck [] k 
response (End)     env = Final Nop
response (m :!: r) env
  = (term m env) <!> (response r env)

(<!>) :: CPSTerm i a -> Answer i a
      -> Answer i a
f <!> Final r      = f r
f <!> Stuck gs x q = Stuck (f : gs) x q
f <!> CoStuck gs q = CoStuck (f : gs) q

term :: (Eq a, Eq i) => Term i a -> CPSEnv i a
     -> CPSTerm i a
term (Var x)   env = case lookup x env of
  Just (CPST m) -> m
  _             -> Stuck [] x
term (Dot m)    env
  = \k -> (term m env) (Arg (term m env) :* k)
term (m :*: n)  env
  = \k -> (term m env) (Arg (term n env) :* k)
term (m :@: i)  env = \k -> (term m env) (i :@ k)
term (Raise)    env = \k -> Final k
term (q :!-> r) env
  = \k -> (response r ((q, CPSQ k) : env))
term (o :?: m)  env
  = \k -> (option o env) k (term m env) k

option :: (Eq i, Eq a) => Option i a -> CPSEnv i a
       -> CPSQuestion i a -> CPSOption i a
option (x :*-> o) env = \q f -> \case
  (y :* k) -> (option o env') q f k
    where env' = (x, CPST (useArg y)) : env
  _        -> f q
option (i :@-> o) env = \q f -> \case
  (j :@ k) | i == j -> (option o env) q f k
  _                 -> f q
option (x :?-> m) env = \_ f -> (term m env')
  where env' = (x, CPST f) : env
\end{haskell}
\caption{Environment and continuation-passing style translation for copatterns with nested options.}
\label{fig:nest-env-cps-code}
\end{figure}

\begin{figure}
\centering
\renewcommand{\basicstylesize}{\footnotesize}
\begin{haskell}
data Answer i a
  = Final   (ClosQuestion i a)
  | Stuck   (MetaCont i a) a (ClosQuestion i a)
  | CoStuck (MetaCont i a) a

type MetaCont i a = [Closure i a]

run :: (Eq a, Eq i) => Response i a -> Answer i a
run r = delim r [] []

eval :: (Eq i, Eq a) => Term i a -> Answer i a
eval m = refocus m [] Nop []

try :: (Eq i, Eq a) => Option i a -> Answer i a
try o = comatch o [] Nop (Raise :/: []) Nop []

delim :: (Eq a, Eq i)
      => Response i a -> ClosEnv i a
      -> MetaCont i a -> Answer i a
delim (Splat k) env (m:/:e:s)
  | Just (QSub q) <- lookup k env
  = refocus m e Nop s
delim (Splat k) env []
  | Just (QSub q) <- lookup k env
  = Final q
delim (Splat k) env s
  = CoStuck s k
delim (End)     env (m:/:e:s)
  = refocus m e Nop s
delim (End)     env []
  = Final Nop
delim (m :!: r) env s
  = delim r env $ (m :/: env) : s

refocus :: (Eq a, Eq i) => Term i a
        -> ClosEnv i a -> ClosQuestion i a
        -> MetaCont i a -> Answer i a
refocus (Var x)    env k s
  | Just (CSub (m:/:e)) <- lookup x env
  = refocus m e k s
refocus (Var x)    env k s
  = Stuck s x k
refocus (Dot m)    env k s
  = refocus m env ((m :/: env) :* k) s
refocus (m :*: n)  env k s
  = refocus m env ((n :/: env) :* k) s
refocus (m :@: i)  env k s
  = refocus m env (i :@ k) s
refocus (Raise)    env k (m:s)
  = refocus (openTerm m) (staticEnv m) k s
refocus (Raise)    env k []
  = Final k
refocus (q :!-> r) env k s
  = delim r ((q, QSub k) : env) s
refocus (o :?: m)  env k s
  = comatch o env k (m :/: env) k s

comatch :: (Eq i, Eq a) => Option i a
        -> ClosEnv i a -> ClosQuestion i a
        -> Closure i a -> ClosQuestion i a
        -> MetaCont i a -> Answer i a
comatch (x :*-> o) env (n:*k) m q s
  = comatch o ((x, CSub n) : env) q m k s
comatch (i :@-> o) env (j:@k) m q s
  | i == j = comatch o env q m k s
comatch (x :?-> n) env k      m _ s
  = refocus n ((x, CSub m) : env) k s
comatch _          env _      m q s
  = refocus (openTerm m) (staticEnv m) q s
\end{haskell}
\caption{Environment-passing, abstract machine interpreter for copatterns with control.}
\label{fig:nest-env-machine-code}
\end{figure}

%% The full-page figures break the flow unless they are put together at the end

\begin{figure*}
\centering
Refocusing / Reduction steps:
\begin{align*}
  \braket{M~N \cmid \sigma \cmid K}
  &\mapsto
  \braket{M \cmid \sigma \cmid N\clos{\sigma} ~ K}
  &
  \braket{M. \cmid \sigma \cmid K}
  &\mapsto
  \braket{M \cmid \sigma \cmid M\clos{\sigma} ~ K}
  &
  \braket{x \cmid \sigma \cmid K}
  &\mapsto
  \braket{M \cmid \sigma' \cmid K}
  \\
  \braket{M~X \cmid \sigma \cmid K}
  &\mapsto
  \braket{M \cmid \sigma \cmid X ~ K}
  &
  \braket{\lambda\{L \to M; \vect{O}\} \cmid \sigma \cmid K}
  &\mapsto
  \braket{L \cmid K \cmid \sigma \cmid M \cmid \vect{O} \cmid \sigma \cmid K}
  &
  &~(\asub{x}{M\clos{\sigma'}} \in \sigma)
\end{align*}

Copattern-matching steps:
\begin{align*}
  \braket{x ~ L \cmid N\clos{\sigma'} ~ K' \cmid \sigma \cmid M \cmid \many{O} \cmid \sigma_0 \cmid K}
  &\mapsto
  \braket{L \cmid K' \cmid \asub{x}{N\clos{\sigma'}},\sigma \cmid M \cmid \many{O} \cmid \sigma_0 \cmid K}
  \\
  \braket{X ~ L \cmid X ~ K' \cmid \sigma \cmid M \cmid \many{O} \cmid \sigma_0 \cmid K}
  &\mapsto
  \braket{L \cmid K' \cmid \sigma \cmid M \cmid \many{O} \cmid \sigma_0 \cmid K}
  \\
  \braket{\varepsilon \cmid K' \cmid \sigma \cmid M \cmid \many{O} \cmid \sigma_0 \cmid K}
  &\mapsto
  \braket{M \cmid \sigma \cmid K'}
  \\
  \braket{L \cmid \varepsilon \cmid \sigma \cmid M \cmid \many{O} \cmid \sigma_0 \cmid K}
  &\not\mapsto
  \qquad\qquad\qquad\qquad\qquad
  (\text{if } L \neq \varepsilon)
  \\
  \braket{L \cmid K' \cmid \sigma \cmid M \cmid \many{O} \cmid \sigma_0 \cmid K}
  &\mapsto
  \braket{\lambda\{\many{O}\} \cmid \sigma_0 \cmid K}
  \qquad(\text{otherwise})
\end{align*}
\caption{Environment-based abstract machine for calculating monolithic copatterns.}
\label{fig:env-machine}
\end{figure*}

\begin{figure*}
\centering

Meta-continuation steps:
\begin{align*}
  \braket{M \ans R \cmid \sigma \cmid S}
  &\mapsto
  \braket{R \cmid \sigma \cmid M\clos{\sigma}; S}
  &
  \braket{\varepsilon \cmid \sigma \cmid M\clos{\sigma'}; S}
  &\mapsto
  \braket{M \cmid \sigma' \cmid \varepsilon \cmid S}
  &
  \braket{q \cmid \sigma \cmid M\clos{\sigma'}; S}
  &\mapsto
  \braket{M \cmid \sigma' \cmid \sigma(q) \cmid S}
\end{align*}

Refocusing / reduction steps:
\begin{align*}
  \braket{M ~ X \cmid \sigma \cmid K \cmid S}
  &\mapsto
  \braket{M \cmid \sigma \cmid X ~ K \cmid S}
  \\
  \braket{M ~ N \cmid \sigma \cmid K \cmid S}
  &\mapsto
  \braket{M \cmid \sigma \cmid N\clos{\sigma} ~ K \cmid S}
  &\qquad
  \braket{\ans q \to R \cmid \sigma \cmid K \cmid S}
  &\mapsto
  \braket{R \cmid \asub{q}{K},\sigma \cmid S}
  \\
  \braket{M. \cmid \sigma \cmid K \cmid S}
  &\mapsto
  \braket{M \cmid \sigma \cmid M\clos{\sigma} ~ K \cmid S}
  &\qquad
  \braket{O \ask M \cmid \sigma \cmid K \cmid S}
  &\mapsto
  \braket{O \cmid \sigma \cmid K \cmid M \clos\sigma \cmid K \cmid S}
\end{align*}

Copattern-matching steps:
\begin{align*}
  \braket{x \to O \cmid \sigma \cmid N\clos{\sigma'}~K \cmid M \clos{\sigma_0} \cmid K_0 \cmid S}
  &\mapsto
  \braket{O \cmid \asub{x}{N\clos{\sigma'}}, \sigma \cmid K \cmid M \clos{\sigma_0} \cmid K_0 \cmid S}
  \\
  \braket{X \to O \cmid \sigma \cmid X~K \cmid M \clos{\sigma_0} \cmid K_0 \cmid S}
  &\mapsto
  \braket{O \cmid \sigma \cmid K \cmid M\clos{\sigma_0} \cmid K_0 \cmid S}
  \\
  \braket{\ask x \to N \cmid \sigma \cmid K \cmid M\clos{\sigma_0} \cmid K_0 \cmid S}
  &\mapsto
  \braket{N \cmid \asub{x}{M\clos{\sigma_0}}, \sigma \cmid K \cmid S}
  \\
  \braket{O \cmid \sigma \cmid K \cmid M\clos{\sigma_0} \cmid K_0 \cmid S}
  &\mapsto
  \braket{M \cmid \sigma_0 \cmid K_0 \cmid S}
  \quad(\text{otherwise})
\end{align*}
\caption{Environment-based abstract machine for controlling compositional copatterns.}
\label{fig:comp-copat-machine}
\end{figure*}

The abstract machines derived in \cref{sec:derive-copat,sec:control-copat} are
based on substitution, which is a correct but notoriously slow implementation of
static binding.  A more efficient implementation technique is to explicitly
thread environments through the machine states and form closures when necessary
to correctly implement static scope.  These kind of environment-based
implementations are standard practice, but correctly managing static scope
through closures can be tricky, its correctness is not as obvious.

In the context of the derivations we have done so far, we could treat addition
of environments and closures to an abstract machine as a complex monolithic
program transformation.  Instead, here we stay within the incremental style, and
perform a small, but obvious transformation at the right level of abstraction
that makes environment-passing straightforward.  Then, turning the crank in the
same way as before will mechanically generate a more efficient abstract machine
with more confidence that it is correct by construction.

\subsection{Closing over monolithic copatterns}
\label{sec:closing-monolithic-copat}

In order to thread environments efficiently, we start from the very beginning
with the small-step semantics.  The main change takes place in the \hs|reduce|
function as shown in \cref{fig:block-env-reduce-code}: the \hs|Redex| it
processes will now contain an explicit environment representing some delayed
substitutions that haven't been finished yet, and its \hs|Reduct| can now return
a \hs|Closure| (pair of an open term and static environment) with potentially
more delayed substitutions.

The reasoning behind why this program transformation is correct with respect to
\cref{fig:block-reduce-code} is that, if we eagerly perform all delayed
substitutions before and after the environment-passing \hs|reduce| step, it is
the same as the substitution-based \hs|reduce| step.  Since \hs|reduce| is a
non-recursive stepping function, this property can be manually confirmed by
manually checking each case.

Since the new \hs|Redex| type now contains closures, we also have to update the
decomposition functions \hs|decomp| and \hs|refocus|.  These now start with
explicit closures and search for the next redex---which follows exactly the
same code structure before, since the search never goes under binders---which
produces a redex with explicit substitution and a question containing closures in
place of raw terms.

Putting this all together, we then get the environment-passing, small step
evaluator \hs|eval| and main driver loop \hs|iter| shown in
\cref{fig:block-env-reduce-code}---already in the in-place refocusing
form---which corresponds to the original small-step interpreter up to performing
the delayed substitutions.  The main correctness property about the top-level
\hs|eval| function can be derived from each step of \hs|iter| by relating the
above relationship of \hs|reduce| and \hs|refocus|.

From here on out, there is nothing new.  Applying the same program
transformations as before---CPS transformation, defunctionalization, loop
fusion, compressing corridor transitions, deforesting, and other
representational data structure changes---yields the environment-passing,
tail-recursive interpreter in \cref{fig:block-env-machine-code}.

We can continue on to derive a continuation-passing style transformation like
before as well, using the same transformation steps---desugaring pattern
matching, $\eta$-reduction, and immediately applying transition functions to all
sub-expressions as soon as they are available.  The resulting code corresponds
to a form of CPS transformation that is parameterized by a static environment
that gets used to interpret both free and bound variables, in the style of many
denotational semantics.  Rephrased as a translation function into the
$\lambda$-calculus, this CPS is as follows:
\begin{itemize}
\item Translating terms $\den{M}[\sigma]$:
\begin{align*}
  \den{x}[\sigma] &= \sigma(x)
  \\
  \den{M~X}[\sigma] &= \lambda k.~ \den{M}[\sigma] ~ (X ~ k)
  \\
  \den{M~N}[\sigma] &= \lambda k.~ \den{M}[\sigma] ~ (\den{N}[\sigma], k)
  \\
  \den{M.}[\sigma] &= \lambda k.~ \den{M}[\sigma] ~ (\den{M}[\sigma], k)
  \\
  \den{\lambda\{\many{O}\}}[\sigma] &= \den{\many{O}}[\sigma]
\end{align*}
\item Translating lists of options $\den{\many{O}}[\sigma]$:
\begin{align*}
  \den{\varepsilon}[\sigma] &= \lambda k.~ k
  \\
  \den{L = M \mid \many{O}}[\sigma]
  &=
  \lambda k.~ \den{L \to M}[\sigma] ~ k ~ \den{\many{O}}[\sigma] ~ k
\end{align*}
\item Translating copattern-matching options $\den{L \to M}[\sigma]$:
\begin{align*}
  \den{\varepsilon \to N}[\sigma] &= \lambda q. \lambda f.~ \den{N}[\sigma]
  \\
  \den{x ~ L \to N}[\sigma] &= \Rec r = \lambda q. \lambda f. \lambda k.
  \\[-1ex]
  &\qquad
  \Case k \Of
  \begin{alignedat}[t]{2}
    &(y, k') &&\to \den{L = N}[\subst{y}{x},\sigma] ~ q ~ f ~ k'
    \\[-1ex]
    &() &&\to r ~ q ~ f
    \\[-1ex]
    &k &&\to f ~ q
  \end{alignedat}
  \\
  \den{X ~ L \to N}[\sigma] &= \Rec r = \lambda q. \lambda f. \lambda k.
  \\
  &\qquad
  \Case k \Of
  \begin{alignedat}[t]{2}
    &(X ~ k') &&\to \den{L = N}[\sigma] ~ q ~ f ~ k'
    \\[-1ex]
    &() &&\to r ~ q ~ f
    \\[-1ex]
    &k &&\to f ~ q
  \end{alignedat}
\end{align*}
\end{itemize}
As a convention, when bound names are introduced on the right-hand side of a
defining equation, they are always chosen to be distinct from the free variables
of $\sigma$ to avoid accidental capture.

\subsection{Closing over compositional copatterns}
\label{sec:closing-compositional-copat}

The refactorings used \cref{sec:refactor} to generalize the calculus for
delimited and compositional copattern matching were orthogonal to the question
about substitution versus environments as the semantics for static variables.
Therefore, we can replay the changes to the environment and continuation-passing
transformation in \cref{sec:closing-monolithic-copat} to derive a similar
environment-based CPS translation of compositional copatterns:
\begin{itemize}
\item Translating responses $\den{R}[\sigma]$
\begin{align*}
  \den{M \ans R}[\sigma]
  &=
  \lambda s.~ \den{R}[\sigma] ~ \lambda q.~ \den{M}[\sigma] ~ q ~ s
  \\
  \den{q}[\sigma] &= \lambda s.~ s ~ \sigma(q)
  \\
  \den{\varepsilon}[\sigma] &= \lambda s.~ s ~ ()
\end{align*}

\item Translating terms $\den{M}[\sigma]$
  \begin{align*}
  \den{x}[\sigma] &= \sigma(x)
  \\
  \den{M~X}[\sigma] &= \lambda k.~ \den{M}[\sigma]~(X~k)
  \\
  \den{M~N}[\sigma] &= \lambda k.~ \den{M}[\sigma]~(\den{N}[\sigma], k)
  \\
  \den{M.}[\sigma] &= \lambda k.~ \den{M}[\sigma]~(\den{M}[\sigma], k)
  \\
  \den{\Raise}[\sigma] &= \lambda k.~ \lambda s. s ~ k
  \\
  \den{O \ask M}[\sigma] &= \lambda k.~ \den{O}[\sigma] ~ \den{M}[\sigma] ~ k
  \\
  \den{\ans q \to R}[\sigma] &= \lambda q.~ \den{R}[\sigma]
\end{align*}

\item Translating options $\den{O}[\sigma]$
  \begin{align*}
  \den{x \to O}[\sigma]
  &=
  \lambda f. \lambda k.
  \begin{alignedat}[t]{2}
    &\Case k \Of
    \\[-1ex]
    &\quad
    (x, k') &&\to \den{O}[\sigma] ~ (\lambda q.~ f~(x, q)) ~ k'
    \\[-1ex]
    &\quad
    k &&\to f ~ k
  \end{alignedat}
  \\
  \den{X \to O}[\sigma]
  &=
  \lambda f. \lambda k.
  \begin{alignedat}[t]{2}
    &\Case k \Of
    \\[-1ex]
    &\quad
    (X~k') &&\to \den{O}[\sigma] ~ (\lambda q.~ f~(X~q)) ~ k'
    \\[-1ex]
    &\quad
    k &&\to f ~ k
  \end{alignedat}
  \\
  \den{\ask x \to M}[\sigma] &= \lambda x.~ \den{M}[\sigma]
\end{align*}
\end{itemize}
The corresponding Haskell embedding is shown in \cref{fig:nest-env-cps-code}.

From here on out, we can turn the CPS transformation into an abstract machine
using the same general derivation technique.  Applying standard code
transformations---defunctionalization, delaying the application of translation
functions until the last moment of application, $\eta$-expansion, and the use of
nested pattern matching---gives the environment-passing, tail-recursive
interpreter shown in \cref{fig:nest-env-machine-code}.

To compare the difference of the low-level execution of the two calculi---one
for monolithic matching of complex copatterns, and the other for compositional
matching of copatterns with control---we can put them in more common forms.
Rephrasing the Haskell implementations as stepping relations on machine
configurations for both calculi are shown in
\cref{fig:env-machine,fig:comp-copat-machine}.

%%% Local Variables:
%%% mode: LaTeX
%%% TeX-master: "derive-copat"
%%% End:

\end{document}